\documentclass[letterpaper,11pt]{article}
\usepackage[utf8]{inputenc}
\usepackage[margin=1in]{geometry}
\usepackage{authblk}

\usepackage{amsmath}
\usepackage{amsthm}
\usepackage{amssymb}
\usepackage{thmtools}
\usepackage{thm-restate}
\usepackage{enumitem}
\usepackage{dsfont}
\usepackage{xspace}
\usepackage{comment}
\usepackage{xparse}
\usepackage[ruled,noend,resetcount,linesnumbered]{algorithm2e}
\usepackage{tikz}
\usepackage[draft]{fixme}
\usepackage[normalem]{ulem}

\usepackage{lineno}
\fxsetup{mode=multiuser,theme=color, layout=inline}
\FXRegisterAuthor{eg}{aeg}{\color{purple}\textbf{Elazar}}
\FXRegisterAuthor{rk}{ark}{\color{red}\textbf{Robi}}
\FXRegisterAuthor{tk}{atk}{\color{blue}\textbf{Tomasz}}
\FXRegisterAuthor{bs}{abs}{\color{violet}\textbf{Barna}}

\usepackage[colorlinks,linkcolor=blue,citecolor=blue,filecolor=blue,urlcolor=blue]{hyperref}
\usepackage[capitalise]{cleveref}

\setitemize{itemsep=0pt,topsep=2pt,parsep=0pt,partopsep=0pt}
\setenumerate{itemsep=0pt,topsep=2pt,parsep=0pt,partopsep=0pt}

\newcommand{\ED}{\mathsf{ED}}

\newcommand{\dd}{\mathinner{.\,.}}
\newcommand{\ceil}[1]{\lceil #1 \rceil}
\newcommand{\floor}[1]{\lfloor #1 \rfloor}
\newcommand{\sub}{\subseteq}
\newcommand{\sm}{\setminus}

\newcommand{\Oh}{\mathcal{O}}
\newcommand{\Ohtilde}{\tilde{\Oh}}
\newcommand{\Thetatilde}{\tilde{\Theta}}
\newcommand{\eps}{\varepsilon}

\renewcommand{\Pr}{\mathbb{P}}
\newcommand{\Exp}{\mathbb{E}}

\newtheorem{theorem}{Theorem}[section]
\newtheorem{corollary}[theorem]{Corollary}
\newtheorem{proposition}[theorem]{Proposition}
\newtheorem{lemma}[theorem]{Lemma}
\newtheorem{fact}[theorem]{Fact}
\newtheorem{claim}[theorem]{Claim}

\newtheorem{problem}[theorem]{Problem}
\theoremstyle{definition}
\newtheorem{definition}[theorem]{Definition}

\newcommand{\fr}[2]{[#1\dd#2)}
\newcommand{\R}{\mathbb{R}}

\newcommand{\Zz}{\mathbb{Z}_{\ge 0}}
\newcommand{\Zp}{\mathbb{Z}_{+}}
\newcommand{\A}{\mathcal{A}}

\newcommand{\fAN}{f_{\scriptscriptstyle\text{AN}}}
\newcommand{\gAN}{g_{\scriptscriptstyle\text{AN}}}

\newcommand{\GED}{\textsc{Gap Edit Distance}\xspace}
\NewDocumentCommand\GEDa{gg}{{\IfNoValueTF{#1}{$(\beta,\alpha)$}{$(#2,#1)$}}-\textsc{Gap Edit Distance}\xspace}
\newcommand{\SED}{\textsc{Shifted Gap Edit Distance}\xspace}
\NewDocumentCommand\SEDa{ggg}{{\IfNoValueTF{#1}{$\beta$}{$#2$}}-\textsc{Shifted} {\IfNoValueTF{#1}{$(\gamma,3\alpha)$}{$(#3,3#1)$}}-\textsc{Gap Edit Distance}\xspace}

\title{Gap Edit Distance via Non-Adaptive Queries: Simple and Optimal}

\author[1]{Elazar Goldenberg}
\author[2]{Tomasz Kociumaka\thanks{Partly supported by NSF CCF grants 1652303 and 1909046, and a HDR TRIPODS Phase II grant 2217058.}}
\author[3]{Robert Krauthgamer\thanks{Partly supported by ONR Award N00014-18-1-2364, the Israel Science Foundation grant \#1086/18, the Weizmann Data Science Research Center, and a Minerva Foundation grant.}}
\author[4]{Barna Saha\protect\footnotemark[1]}

\affil[1]{The Academic College of Tel Aviv-Yaffo, Israel}
\affil[ ]{\texttt{elazargo@mta.ac.il}}
\affil[2]{Max Planck Institute for Informatics, Germany}
\affil[ ]{\texttt{tomasz.kociumaka@mpi-inf.mpg.de}}
\affil[3]{Weizmann Institute of Science, Israel}
\affil[ ]{\texttt{robert.krauthgamer@weizmann.ac.il}}
\affil[4]{University of California, San Diego, United States}
\affil[ ]{\texttt{barnas@ucsd.edu}}

\date{}

\begin{document}
\maketitle

\begin{abstract}
  We study the problem of approximating edit distance in sublinear time.
  This is formalized as the \GEDa{k^c}{k} problem, 
  where the input is a pair of strings $X,Y$ and parameters $k,c>1$,
  and the goal is to return YES if $\ED(X,Y)\leq k$, NO if $\ED(X,Y)> k^c$,
  and an arbitrary answer when $k < \ED(X,Y) \le k^c$.
  Recent years have witnessed significant interest in designing sublinear-time algorithms for \GED.
  
  In this work, we resolve the non-adaptive query complexity of \GED for the entire range of parameters, 
  improving over a sequence of previous results. 
  Specifically, we design a non-adaptive algorithm with query complexity $\Ohtilde(n/k^{c-0.5})$,
  and we further prove that this bound is optimal up to polylogarithmic factors.
  
  Our algorithm also achieves optimal time complexity $\Ohtilde(n/k^{c-0.5})$ whenever $c\geq 1.5$. 
  For $1<c<1.5$, the running time of our algorithm is $\Ohtilde(n/k^{2c-2})$.
  In the restricted case of $k^c=\Omega(n)$, this matches a known result 
  [Batu, Ergün, Kilian, Magen, Raskhodnikova, Rubinfeld, and Sami; STOC 2003],
  and in all other (nontrivial) cases, our running time is strictly better 
  than all previous algorithms, including the adaptive ones.
  However, independent work of Bringmann, Cassis, Fischer, and Nakos [STOC 2022] provides an adaptive algorithm that bypasses the non-adaptive lower bound, but only for small enough~$k$ and $c$.
\end{abstract}

\section{Introduction}
The \textit{edit distance} is a ubiquitous distance measure on strings. It finds applications in various fields including computational biology, pattern recognition, text processing, information retrieval, and many more. 
The edit distance between strings $X$ and 
$Y$, denoted by $\ED(X,Y)$, is defined as the minimum number of character insertions, deletions, and substitutions needed to convert $X$ into~$Y$. A simple textbook dynamic programming computes edit distance in quadratic time. Moreover, under reasonable hardness assumptions, such as the Strong Exponential-Time Hypothesis, no truly subquadratic-time algorithm for this problem exists~\cite{ABW15,BK15,AHWW16,BI18}.

When dealing with enormous amounts of data, such as DNA strings, big data storage, etc., quadratic running time might be prohibitive,
leading a quest for faster algorithms that find an approximate solution.
A long line of research towards that goal~\cite{BEKMRRS03,BJKK04,BES06,AKO10,AO12,CDGKS18,BEGHS18} recently culminated with
an almost-linear-time approximation algorithm by Andoni and Nosatzki~\cite{AN20} that, for any desired $\eps>0$,
runs in $O(n^{1+\eps})$ time and achieves an approximation factor that depends only on $\eps$,
that is, a constant-factor approximation for any fixed $\eps > 0$.

The growing interest in modern computational paradigms,
like streaming and sketching (sublinear space), sampling and property testing (sublinear time), and massively parallel computation, 
sparked interest in \emph{sublinear-time} algorithms.
It started with a seminal work of Batu, Ergün, Kilian, Magen, Raskhodnikova, Rubinfeld, and Sami~\cite{BEKMRRS03}, 
and developed into an exciting sequence of results on approximating $\ED$ in sublinear time~\cite{AO12,GKS19,BCR20,KS20a,BCFN22}. 
(These are sometimes called estimation algorithms to emphasize that they 
approximate $\ED$ without necessarily constructing a witness alignment.) 

A sublinear-time algorithm for $\ED$ cannot be expected to attain constant-factor approximation, since even in the case where the edit distance is $O(1)$, a linear fraction of the input strings must be queried.
Hence, the aim here is to solve the promise problem \GEDa{k^c}{k},
which asks to return YES if $\ED(X,Y)\leq k$, NO if $\ED(X,Y)> k^c$, and an arbitrary answer otherwise. 
The accuracy of the aforementioned results depends on gap ``size'' $c$
and the gap ``location''~$k$;
their performance is measured in terms of their query and time complexity,
and also qualitatively whether they query the input strings adaptively 
(i.e., each query may depend on the results of earlier queries).

The first contribution~\cite{BEKMRRS03} addressed the case $k^c=\Omega(n)$. Under this restriction, they obtained a sublinear-time algorithm that runs in~$\Ohtilde(k^2/n+\sqrt{k})$ time.\footnote{The $\Ohtilde(\cdot)$ notation hides factors polylogarithmic in $n$.}
Moreover, they showed a query-complexity lower bound of $\Omega(\sqrt{k})$, 
rendering their result optimal for $c\ge 1.5$. 
Andoni and Onak~\cite{AO12} were the first to overcome the limitation that $k^c=\Omega(n)$.
However, their query complexity $\hat{\Oh}(n^{2}/k^{2c-1})$
reduces to $\hat{\Oh}(k)$ when $k^c=\Omega(n)$, far above that of~\cite{BEKMRRS03}.\footnote{The $\hat{\Oh}(\cdot)$ notation hides factors subpolynomial in $n$.}
Interestingly, both these algorithms are non-adaptive. 

In recent years, further progress has been achieved, 
mostly by exploiting adaptive queries,
particularly by Goldenberg, Krauthgamer, and Saha~\cite{GKS19}, 
and subsequently by Kociumaka and Saha~\cite{KS20a}, 
who improved over~\cite{AO12} when $k$ is small. 
The running time $\Ohtilde(n/k^{c-1}+k^3)$ of~\cite{GKS19} 
had an undesirable cubic dependency on $k$, 
which was improved to quadratic in~\cite{KS20a} at the price of an extra $\Ohtilde(k^{2.5-c}n^{0.5})$ term appearing for $c< 2$. 
%There are simpler non-adaptive algorithms that have better polynomial dependency on $k$ but comes with a significant worse complexity in terms of $n$ (e.g. from $\frac{n}{k}$ to $\frac{n}{\sqrt{k}}$ for $c=2$)~\cite{BCR20}. 
%also provided improvements in the polynomial dependency on $k$.
%Their algorithm comes with a worse complexity in terms of $n$
%(for example when $c=2$, it worsened from $\frac{n}{k}$ to $\frac{n}{\sqrt{k}}$),
%but it is non-adaptive and simple.
Non-adaptive algorithms often tend to be simple, but they are generally less powerful. 
So far, the best results in the regime of small $k$ came through carefully using adaptive queries~\cite{GKS19,KS20a}. 
Hence, it seemed plausible that adaptivity would be crucial 
to improving beyond~\cite{AO12} for large $k$ as well.
%However, their result did not generalize to $k^c=o(n)$. Andoni and Onak overcame the restriction and proposed an algorithm with a running time of $\Ohtilde(\frac{n^{2+o(1)}}{k^{2c-1}})$~\cite{AO12}. The algorithm works well when $k$ is relatively large, and its performance degrades considerably as $k$ becomes small. In contrast, a recent result by Goldenberg, Krauthgamer, and Saha showed a running time of $\Ohtilde(\frac{n}{k^{c-1}}+k^3)$ which works well for the small $k$ regime~\cite{GKS19}. A comparative analysis of these results are shown in Figure~\ref{fig:good} for $c=2$. The dependency of $k^3$ is naturally unsatisfactory. In the light of the well-known exact algorithm for $\ED$ in $O(n+k^2)$~\cite{LMS98,LV88}, the question remained if the dependency on $k$ can be reduced. A subsequent result by Kociumaka and Saha improved the dependency from $k^3$ to $k^2$~\cite{KS20a}, whereas  Brankensiek, Charikar and Rubinstein eliminated the dependency on $k$ by significantly increasing the dependency on $n$~\cite{BCR20}. Specifically, for $c=2$, they respectively obtained  running times of $\Ohtilde(\frac{n}{k}+k^2)$, and $\Ohtilde(\frac{n}{\sqrt{k}})$. In absence of a lower bound, it remained unclear whether any of these results are tight! Moreover, the first two results~\cite{BEKMRRS03,AO12} use non-adaptive queries, whereas the rest exploited adaptive queries, i.e., the algorithms cannot issue all the queries simultaneously, and rather must use the results of earlier queries to decide the next set of queries and so on. 

\paragraph*{Technical Contributions}
In light of prior work, the following main questions remained open. 
\begin{itemize}
    \item Can we remove/reduce the polynomial dependency on $k$ from~\cite{GKS19, KS20a} without degrading the dependency on $n$? 
    \item Is adaptivity needed to achieve complexity $\Ohtilde(n/k^{c-1})$ for small $k$?
    \item Can we obtain tight query-complexity lower bounds?
\end{itemize}

\begin{table}[t]
	\centering
  \def\arraystretch{1.25}
	\begin{tabular}{|l|l|l|l|}
          \hline
          Time Complexity &  Non-Adaptive & Restrictions & Reference   \\
          \hline\hline
          $\Ohtilde\left(k^{0.5}\right)=\Ohtilde\left(n/k^{c-0.5}\right)$ & Yes    & $k^c=\Omega(n)$, $c\ge 1.5$  & \cite{BEKMRRS03}           \\
          $\Ohtilde\left(k^2/n\right)=\Ohtilde\left(n/k^{2c-2}\right)$ & Yes    & $k^c=\Omega(n)$, $c<1.5$  & \cite{BEKMRRS03}      \\
          \hline
          $\hat{\Oh}({n^{2}}/{k^{2c-1}})$ & Yes    & &\cite{AO12} \\
          \hline          
          $\Ohtilde(n/k^{c-1}+k^3)$ & No    & & \cite{GKS19} \\
          $\Ohtilde(n/k^{c-1.5}+k^{2.5-c})$ & Yes    & $c\ge 1.5$ & \cite{BCR20} \\
          $\Ohtilde(n/k^{c-1}+k^2+\sqrt{n}\cdot k^{2.5-c})$ & No    & & \cite{KS20a} \\
          \hline
          $\hat{\Oh}(n/{k^c}+n^{0.8}+k^4)$ & No   &  & \cite{BCFN22} \\
          \hline
          $\hat{\Oh}(n/k^{c-1})$ & Yes    &  & \cref{cor:simple} \\
          $\Ohtilde(n/k^{c-0.5})$ & Yes    & $c\ge  1.5$ & \cref{thm:main} \\
          $\Ohtilde(n/k^{2c-2})^{*}$ & Yes    & $c< 1.5$ & \cref{thm:main} \\
          \hline
	\end{tabular}
	
  $^{*}$Query complexity is $\Ohtilde(n/k^{c-0.5})$, lower than the time complexity.
  \caption{Sublinear-time algorithms for \GEDa{k^c}{k}.}%
	\label{TaxomonyTable}
\end{table}

In~\cref{sec:simpleAlg}, we present a simple \emph{non-adaptive} algorithm 
that removes the polynomial dependency on $k$ entirely, 
thus answering the first two questions. 
The algorithm solves the \GEDa{k^c}{k} problem with time complexity $\hat{\Oh}(n/k^{c-1})$, as follows.

\begin{theorem}[Simplified version of \cref{cor:simple}]\label{thm:1}
  For every constant $c>1$, there is a non-adaptive randomized algorithm
  that solves \GEDa{k^c}{k} in time $\hat{\Oh}(n/k^{c-1})$.
\end{theorem}

This result already improves upon all prior results~\cite{AO12,GKS19,KS20a,BCR20},
except for the earliest algorithm of~\cite{BEKMRRS03} that applies only for $k^c=\Omega(n)$ (see below for a comparison with independent work~\cite{BCFN22}).
The algorithm abandons the recent approach of~\cite{GKS19,KS20a} of scanning the two input strings and tracking their periodicity structure using adaptive queries. 
Instead, our baseline is Andoni and Onak's algorithm~\cite{AO12},
which samples a few blocks of predetermined length from each string, 
and computes only a ``local'' alignment
(between the $i$th block of $X$ and the $i$th block of $Y$).
In~\cite{AO12}, the block length is optimized according to the gap parameters $k$ and $c$.
Somewhat surprisingly, just by sampling blocks of different lengths and using all of them simultaneously (instead of choosing a single block length), 
we achieve a significantly better result.
The details, including a technical overview, appear in~\cref{sec:simpleAlg}.

Our main result still uses non-adaptive sampling
and achieves a significant improvement for the entire range of $c$,
and in particular generalizes or improves upon all previous bounds.
By building on the above simple algorithm, 
we first improve the query complexity (in \cref{sec:qc})
and then also the time complexity (in \cref{sec:faster}),
both in the polynomial dependency on $k$ and in the $n^{o(1)}$-factor. 

\begin{theorem}[Simplified version\protect\footnote{It suffices to use \cref{thm:main} with any constant $h > \frac{1}{1-c}$, and resort to an exact algorithm if $k = \Ohtilde(1)$. We remark that \cref{thm:main} additionally has some limited applicability to $c=1+o(1)$
and can solve \GEDa{k\cdot 2^{\Oh(\sqrt{\log n})}}{k} using $\hat{\Oh}(n/\sqrt{k})$ queries.} of \cref{thm:main}]\label{thm:simplified}
  For every constant $c>1$, there is a non-adaptive randomized algorithm
  that solves \GEDa{k^c}{k}
  using $\Ohtilde(n/k^{c-0.5})$ queries and $\Ohtilde(n/k^{\min(c-0.5,2c-2)})$ time.
\end{theorem}

Our final contribution is a new lower bound for non-adaptive algorithms
that applies for all values of $k$ and $c$ (see \cref{sec:lb}).
It extends a previous lower bound of~\cite{BEKMRRS03},
which handles only the very special case $k^c=\Omega(n)$.

\begin{theorem}[Simplified version of Theorem~\ref{thm:lb}]
For every constant $c >1$ and parameters $n\ge k\ge 1$,
    every non-adaptive algorithm solving %all instances of 
    the \GEDa{k^c}{k} problem 
    has expected query complexity $\Omega(n/k^{c-0.5})$.
%    or error probability exceeding $\frac13$.
\end{theorem}

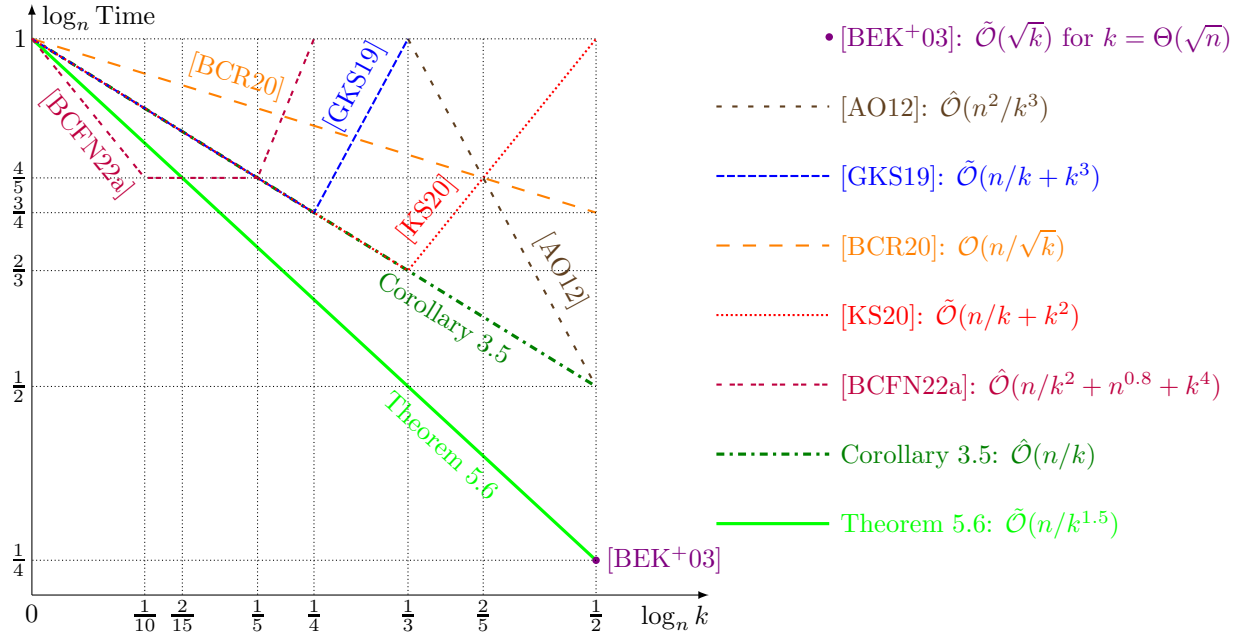
\begin{figure}[t]
  \begin{center}
\hypersetup{hidelinks}{
\begin{tikzpicture}[x=15cm,y=9.25cm,every node/.style={font=\small}]
    \draw[-latex] (0,0.2) -- (0.57, 0.2) node[below]{$\log_n k$} -- (0.6, 0.2);
    \draw[-latex] (0,0.2) -- (0, 1.03) node[right]{$\log_n \text{Time}$} -- (0,1.05);

    \draw(0,1) node[left=-2]{$1$};
    \draw(0,1/4) node[left=-2]{$\frac14$};
    \draw(0,1/2) node[left=-2]{$\frac12$};

    \draw(0,3/4) node[left=-2]{$\frac34$};
    \draw(0,2/3) node[left=-2]{$\frac23$};
    \draw(0,4/5) node[left=-2]{$\frac45$};

    \draw[thin, densely dotted] (0,1) -- (1/2,1) -- (1/2,0.2);

    \draw[very thick,dash pattern=on 3pt off 2pt on \the\pgflinewidth off 2pt,green!50!black](0,1) --node[pos=0.75,sloped,below]{\cref{cor:simple}} (1/2,1/2);

    \draw[very thick,green,green](0,1) --node[pos=0.75,sloped,below]{\cref{thm:main}} (1/2,1/4);

    \draw(1/2, 0.2) node[below=-2]{$\frac12$};
    \draw(0, 0.2) node[below]{$0$};
    \draw(1/4, 0.2) node[below=-2]{$\frac14$};
    \draw(1/3, 0.2) node[below=-2]{$\frac13$};
    \draw(2/5, 0.2) node[below=-2]{$\frac25$};
    \draw(1/10, 0.2) node[below=-2]{$\frac1{10}$};
    \draw(2/15, 0.2) node[below=-2]{$\frac2{15}$};
    \draw(1/5, 0.2) node[below=-2]{$\frac15$};

    \draw[violet](1/2,1/4) node[circle,fill,inner sep=1pt]{} node[right]{\cite{BEKMRRS03}};

    \draw[thin, densely dotted] (1/2,1/4) -- (0,1/4);

    \draw[brown!50!black,thick,dash pattern=on 2pt off 4pt](1/3,1) --node[pos=0.7,sloped,above]{\cite{AO12}} (1/2,1/2);
    \draw[thin, densely dotted] (1/2,0.5) -- (0,0.5);
    \draw[thin, densely dotted] (1/3,1) -- (1/3,0.2);

    \draw[thick,orange,dash pattern=on 5pt off 5pt](0,1) -- node[pos=0.35,sloped,above]{\cite{BCR20}} (1/2,3/4);

    \draw[thick,blue,dash pattern=on 3pt off 1pt](0,1) -- (1/4, 3/4) -- node[pos=0.66,sloped,above]{\cite{GKS19}} (1/3,1);
    \draw[thin, densely dotted] (1/2,3/4) -- (0,3/4);
    \draw[thin, densely dotted] (1/4,1) -- (1/4,0.2);

    \draw[thick, red,dash pattern=on \the\pgflinewidth off 1pt](0,1) -- (1/3, 2/3) --node[pos=0.2,sloped,above]{\cite{KS20a}} (1/2,1);
    \draw[thin, densely dotted] (1/2,2/3) -- (0,2/3);
    \draw[thin, densely dotted] (1/2,4/5) -- (0,4/5);
    \draw[thin, densely dotted] (2/5,1) -- (2/5,0.2);

    \draw[thick,purple,dash pattern=on 2pt off 3pt on 2pt off 1pt](0,1) -- node[pos=0.65,sloped,below]{\cite{BCFN22}} (1/10,4/5)
    -- (1/5,4/5) -- (1/4,1);
    \draw[thin, densely dotted] (1/10,1) -- (1/10,0.2);
    \draw[thin, densely dotted] (2/15,1) -- (2/15,0.2);
    \draw[thin, densely dotted] (1/5,1) -- (1/5,0.2);

    \begin{scope}[xshift=.85cm,yshift=-8.3cm,yscale=2]
      \draw[very thick,dash pattern=on 3pt off 2pt on \the\pgflinewidth off 2pt,green!50!black] (0.55,0.65) -- (0.65,0.65) node[right]{\cref{cor:simple}: $\hat{\Oh}(n/k)$};
      \draw[very thick,green,green] (0.55,0.6) -- (0.65,0.6) node[right]{\cref{thm:main}: $\Ohtilde(n/k^{1.5})$};

        \draw[thick,purple,dash pattern=on 2pt off 3pt on 2pt off 1pt] (0.55,0.7) -- (0.65,0.7) node[right]{\cite{BCFN22}: $\hat{\Oh}(n/k^2+n^{0.8}+k^{4})$};

        \draw[violet](0.65,0.95) node[circle,fill,inner sep=1pt]{} node[right]{\cite{BEKMRRS03}: $\Ohtilde(\sqrt{k})$ for $k=\Theta(\sqrt{n})$};
        \draw[brown!50!black,thick,dash pattern=on 2pt off 4pt] (0.55,0.9) -- (0.65,0.9) node[right]{\cite{AO12}: $\hat{\Oh}(n^{2}/k^3)$};
        \draw[thick,orange,dash pattern=on 5pt off 5pt] (0.55,0.8) -- (0.65,0.8) node[right]{\cite{BCR20}: $\Oh(n/\sqrt{k})$};
        \draw[thick,blue,dash pattern=on 3pt off 1pt] (0.55,0.85) -- (0.65,0.85) node[right]{\cite{GKS19}: $\Ohtilde(n/k+k^3)$};
        \draw[thick, red,dash pattern=on \the\pgflinewidth off 1pt] (0.55,0.75) -- (0.65,0.75) node[right]{\cite{KS20a}: $\Ohtilde(n/k+k^2)$};
    \end{scope}

  \end{tikzpicture}
}
\end{center}
 \caption{The running times of algorithms for \GEDa{k^2}{k}.}%
 \label{fig:1}
\end{figure}

Altogether, we obtain optimal non-adaptive query complexity for all $c>1$, 
and furthermore optimal time complexity for a large regime (all $c \geq 1.5$). 
In particular, we achieve optimal query and time complexity 
for the quadratic gap edit distance problem ($c=2$),
which was the focus of all recent work~\cite{GKS19,KS20a,BCR20}. 
When $c< 1.5$, we match the time bound of~\cite{BEKMRRS03} 
and further remove their restriction that $k^c=\Omega(n)$.

Table~\ref{TaxomonyTable} lists all the known algorithmic bounds, 
including our, previous, and independent results.
It is instructive to compare them against $\Omega(n/k^{c-0.5})$,
our (tight) lower bound for non-adaptive algorithms (\cref{thm:lb}). 
Figure~\ref{fig:1} plots these bounds for a quadratic gap ($c=2$).
One can see that our main result, \cref{thm:main}, 
improves over all previous results for the entire range of $k$,
although independent work~\cite{BCFN22} provides a further improvement for small $k$ 
by using adaptive sampling and thus bypassing our lower bound.

\paragraph*{Open Questions}
Our results completely resolve the non-adaptive query complexity of \GED. 
The lower bound of~\cite{BEKMRRS03} applies to adaptive queries as well,
and we match this lower bound using a non-adaptive algorithm. 
Hence, adaptivity cannot help at the extreme regime of $k^c=\Omega(n)$. 
The question remains though whether adaptivity is useful to improve the complexity further in the intermediate regime,
where the $\hat{\Oh}(k^{4})$ term in the running time of~\cite{BCFN22} makes that solution slower than ours.

Another open problem is to improve the time complexity (ideally to match the query-complexity lower bound) for $1<c<1.5$ or to strengthen the lower bound for time complexity, showing a separation between time and query complexity (as in~\cite{DBLP:journals/jcss/AlonVKK03} for the max-cut problem, for example).

\paragraph*{Related Work} %\label{sec:related}

Sublinear-time algorithms were studied for several related string problems, 
including the Ulam metric~\cite{AN10,NSS17},
longest increasing subsequence (LIS)~\cite{SS17,MS21},
and shift finding~\cite{AHIK13}. 
There are also sublinear-space streaming algorithms for edit distance~\cite{GJKK07,GG10,SS13,EJ08,CGK16,BZ16,CFHJ0RSZ21}. 
Several algorithms for edit distance leverage preprocessing (of one or both strings independently) to perform further computations in sublinear-time~\cite{AKO10,GRS20,BCR20,BCFN22b}.

\paragraph*{Non-Adaptive Sampling}
Our algorithms only require a non-adaptive sampling. %(i.e., they read the input strings at positions that are picked non-adaptively, and can thus be chosen in advance, albeit coordinated between the two strings).
While these might bring inferior performance (running time or query complexity)
compared to algorithms using adaptive sampling,
as indeed obtained independently of our work in~\cite{BCFN22} for small values of~$k$, %Bringmann-Cassis-Fischer-Nakos
numerous applications can gain from --- or even require --- non-adaptive sampling.
Consider for example a distributed setting where the input $XY$ is partitioned
into $p\ge 2$ substrings, held by distinct players 
that communicate in the blackboard model (equivalent to a broadcast channel).
One particular case of interest is two players, one holding $X$ and the other holding~$Y$.
Every sampling algorithm $\A$ has an obvious distributed implementation 
whose communication complexity is precisely the query complexity of $\A$,
but clearly a non-adaptive $\A$ requires only one round of communication
(assuming shared randomness).
For another example, consider $t\ge 3$ input strings $X^{(1)},\ldots,X^{(t)}$
and a goal of estimating the edit distance between every pair of strings.
When implementing a non-adaptive sampling algorithm $\A$,
it suffices to sample each string $X^{(i)}$ only once, 
and use the sample across all the $t-1$ executions involving the string $X^{(i)}$,
thereby running $O(t^2)$ executions of $\A$ using only $O(t)$ sets of samples,
reducing communication by factor $t$ compared to using adaptive sampling.

\section{Preliminaries}%
\label{prelims}

\begin{fact}\label{fct:hered}
    Let $X,Y\in \Sigma^n$. For every $i,j\in [0\dd n]$ with $i\le j$, we have $\ED(X\fr{i}{j},Y\fr{i}{j}) \le \ED(X,Y)$.
    \end{fact}
    
    \begin{fact}\label{fct:subadd}
    For all strings $X_1,X_2,Y_1,Y_2\in \Sigma^*$, we have $\ED(X_1X_2,Y_1Y_2)\le \ED(X_1,Y_1)+\ED(X_2,Y_2)$.
    \end{fact}

\begin{problem}[\GEDa]
      Given strings $X,Y\in \Sigma^n$ and integers $\alpha \ge \beta \ge 0$,
      return YES if $\ED(X,Y)\le \beta$, NO if $\ED(X,Y)>\alpha$, and an arbitrary answer otherwise.
\end{problem}

\begin{theorem}[Landau and Vishkin~\cite{LV88}]\label{thm:lv}
There exists a deterministic algorithm that solves any instance of the \GEDa problem
(with arbitrary $\alpha\ge \beta \ge 0$) in $\Oh(n+\beta^2)$ time.
\end{theorem}

\begin{fact}[see e.g.~\cite{HIM12,KOR00}]\label{fct:hm}
There exists a randomized algorithm that solves any instance of the \GEDa{\alpha}{0} problem
in $\Oh(\frac{n}{1+\alpha})$ time with success probability at least $\frac23$.
\end{fact}

\begin{theorem}[Andoni and Nosatzki~\cite{AN20}]\label{thm:an}
There exist decreasing functions $\fAN,\gAN:\R_+\to \R_{\ge 1}$ and a randomized algorithm $\A$ that, given $X,Y\in \Sigma^n$ and $\eps \in \R_+$, runs in $\Oh(\gAN(\eps)n^{1+\eps})$ time and returns a value $\A(X,Y,\eps)$ satisfying \[\Pr[\ED(X,Y)\le \A(X,Y,\eps) \le \fAN(\eps) \ED(X,Y)]\ge\tfrac23.\]
\end{theorem}

Below, $\fAN$ and $\gAN$ denote the functions of \cref{thm:an}.

\begin{corollary}\label{cor:an}
There exists a randomized algorithm that, given $\eps,\delta \in \R_+$ and an instance
of \GEDa satisfying $\alpha \ge \floor{\fAN(\eps)\beta}$, solves the instance in time $\Oh(\gAN(\eps)n^{1+\eps}\log\frac1\delta)$ with error probability at most $\delta$.
\end{corollary}
\begin{proof}
Consider running the algorithm of \cref{thm:an}.
If $\ED(X,Y)\le \beta$, then the answer is at most $\fAN(\eps)\beta < \alpha+1$ with probability at least~$\frac23$.
If $\ED(X,Y) > \alpha$, then the answer is at least $\alpha+1$ with probability at least $\frac23$.
Hence, comparing the answer against $\alpha+1$ solves the gap problem
in time $\Oh(\gAN(\eps)n^{1+\eps})$ with success probability at least $\frac23$.
The success probability can be amplified to at least $1-\delta$ by running the algorithm $\Theta(\log \frac1\delta)$ times
with independent randomness and returning the dominant answer.
\end{proof}

\section{Simple Algorithm}\label{sec:simpleAlg}

The main result of this section is a randomized algorithm for the \GEDa problem that,
under mild technical conditions, 
% $\alpha \ge \floor{\fAN(\eps)\beta}$,
% \rknote{For simplicitly let's omit discussion of running time and treat $\delta$ as a fixed constant}
makes $\hat{\Oh}(\frac{\beta}{\alpha} \cdot n)$
non-adaptive queries to the two input strings.
The precise time bound depends on the functions $\fAN$ and $\gAN$ from \cref{thm:an}
(see \cref{cor:our} for the formal statement; here, we assumed fixed $\eps,\delta>0$).
Our algorithm is essentially a reduction (presented in \cref{sec:simpleReduction})
to the same gap problem but with smaller gap parameters,
building upon an earlier reduction of Andoni and Onak~\cite{AO12}.%
\footnote{The reduction in~\cite{AO12} is presented 
  as an application of their almost-linear-time approximation algorithm.
}
In a nutshell, these reductions partition the two input strings into blocks
and call an oracle that solves gap problems 
on a few randomly chosen block pairs. 
The key difference from~\cite{AO12} is that
their reduction uses one block length,
while ours essentially uses all feasible block lengths. 

We start below with an overview of both reductions (\cref{sec:SimpleOverview}),
followed by a quick proof of their reduction (\cref{sec:AO}),
which makes it easier to read our reduction (\cref{sec:simpleReduction}) 
and also to compare the two.
To obtain our final result, we only need to implement the oracle,
and we simply plug in the state-of-the-art almost-linear-time algorithm of~\cite{AN20} 
into our reduction (\cref{sec:simpleAlgCor}).

\subsection{Overview}%
\label{sec:SimpleOverview}

To simplify this overview, we shall assume an algorithm
that approximates the edit distance within factor $f=\hat{\Oh}(1)$ in time $\hat{\Oh}(n)$,
and we shall refer to it as an oracle that solves \GEDa
in almost-linear time whenever $\alpha \ge f \beta$. 
Such algorithms were devised in~\cite{AO12,AN20},
and their precise bounds are not important for this overview.

We first sketch the algorithm (reduction) of Andoni and Onak~\cite{AO12}.
It partitions the two input strings $X,Y$
into $m:={\frac{n}{b}}$ blocks of length $b$ that will be determined later,
denoting their respective $i$th blocks by $X_i$ and $Y_i$ for $i\in [0\dd m)$.
If the algorithm determines that
$\ED(X_i,Y_i)>\beta$ for some $i\in [0\dd m)$,
then, by \cref{fct:hered}, also $\ED(X,Y)>\beta$, 
and the algorithm is safe to return NO.\@
The algorithm's strategy is just to search for such a ``NO witness'';
for this, it samples several indices~$i$,
calls the oracle to solve \GEDa{f\beta}{\beta} on the corresponding pairs $(X_i,Y_i)$,
%where $\phi$ is a parameter to be determined,
and returns NO if and only if at least one of the oracle calls returned NO.\@
This algorithm is clearly correct whenever $\ED(X,Y)\leq\beta$, 
so we only need to consider $\ED(X,Y)>\alpha$.
In that case, by \cref{fct:subadd},
$\ED(X_i,Y_i) > \frac{\alpha}{m}$ holds for an average $i$
(a crude intuition is that an average block ``contains'' many edit operations).
For this sketch, let us consider only the two extreme scenarios. 
In one scenario, $\ED(X_i,Y_i)$ has the same value for all $i$;
if $\frac{\alpha}{m} \geq f\beta$,
then, no matter which block our algorithm samples,
the oracle will return NO on it, and our algorithm will also return NO.\@
We will thus constrain our choice of $m$ to satisfy $\frac{\alpha}{m} \geq f\beta$. 
In the other extreme scenario,
$\ED(X_i,Y_i)$ has a large value for a few indices $i$
and a small value, say zero for simplicity, for all other indices.
These large values are bounded by $\ED(X_i,Y_i)\leq b$;
hence, the first group must contain at least $\frac{\alpha}{b}$ indices $i$
(again by \cref{fct:subadd}). 
To have a good chance of sampling at least one of these indices,
our algorithm should sample each $i$ with probability (at least) 
$\rho = \Omega(\frac{b}{\alpha})$.
To optimize algorithm's query complexity,
we set the parameters to minimize the sampling rate $\rho$,
i.e., minimize $b$ or, equivalently, maximize~$m$.
Due to the constraint from above, the optimal choice is thus
% $m=\Theta(\alpha/(f\beta))$.
$b=\frac{n}{m} = \frac{n\cdot f\beta}{\alpha}$. 
The query complexity of this algorithm is 
$\Oh(\rho n)=\Oh(\frac{b}{\alpha} \cdot n)
% = O(\frac{n^2\cdot f\beta}{\alpha^2})
= \hat{\Oh}(\frac{n^2\cdot \beta}{\alpha^2})$,  
and the running time is almost-linear in the query complexity,
and thus bounded similarly.

The true limitation of this approach is that it uses a single block length $b$.
It is somewhat hidden because
we compare $\ED(X_i,Y_i)$ only against the natural threshold $\beta$
(and $f\beta$, but the factor $f$ is almost negligible here),
which leads to an optimal choice of $b$.
One idea is to use a different block length~$b$, or even multiple lengths.
But should it be larger or smaller? And what advantage can we gain from it?

What turns out to work well is a multi-level approach,
which partitions the input strings into blocks of different lengths (all powers of $2$)
and samples blocks from all the levels at the same rate $\rho$.
The query complexity is $O(\rho n)$ for each level,
and there are only $O(\log n)$ levels,
but we can now use sampling rate $\rho=\Theta(\frac{f\beta}{\alpha})$, 
which is significantly lower than $\Theta(\frac{nf\beta}{\alpha^2})$ needed in~\cite{AO12}. 
To understand this improvement in the sampling rate, 
recall the two extreme scenarios mentioned above. 
In the first scenario, where edits are spread evenly among the length-$b$ blocks for $b=\Omega(\frac{n\cdot f\beta}{\alpha})$,
we already argued that querying any length-$b$ block suffices to detect a ``NO witness''. 
In the second scenario, consider shorter blocks of length $\Oh(f \beta)$,
and suppose that each $\ED(X_i,Y_i)$ is either zero or exceeds $f\beta$. 
The number of indices $i$ in the latter group must be at least 
$\frac{\alpha}{\Oh(f \beta)}$ (again by \cref{fct:subadd}),
and they are all ``NO witnesses''. 
To have a good chance of sampling at least one of them, 
it suffices to use rate $\rho=\Theta(\frac{f\beta}{\alpha})$.

Our proof considers all levels and,
for each position $j\in [1\dd n]$,  identifies a ``suitable'' level
based on the distribution of errors in the proximity of $j$. 
This is achieved by decomposing $[1\dd n]$ into blocks of varying sizes,
so that the block covering position $j$ reveals the level responsible for~$j$.
Perhaps surprisingly, the analysis is thus adaptive 
even though the algorithm only makes non-adaptive queries!

\subsection{The Reduction of Andoni and Onak}\label{sec:AO}

We recall a sublinear-time algorithm of Andoni and Onak \cite[Section 4]{AO12}
that makes calls to an oracle solving the same gap problem but with a smaller gap.
They implement this oracle using their main result, 
which is an almost-linear-time approximation algorithm. 
We review their proof to illustrate how our algorithm and analysis are different. 

\begin{theorem}\label{thm:ao}
There exists a randomized reduction that, given a parameter $\phi\in \Zp$
and an instance of \GEDa satisfying $\frac13\alpha\ge \phi \ge \beta \ge 1$,
solves the instance using $\Oh(\frac{n}{\alpha})$
non-adaptive calls to an oracle for \GEDa{\phi}{\beta}
involving substrings of total length $\Oh(\frac{\phi n^2}{\alpha^2})$.
The reduction takes $\Oh(\frac{n}{\alpha})$ time,
does not access the input strings,
and errs with probability at most~$\frac{1}{e}$.
\end{theorem}

\begin{proof}
Let us partition the input strings $X,Y$ into $m:=\ceil{\frac{n}{b}}$ \emph{blocks} of length $b:=\ceil{\frac{3\phi n}{\alpha}}$
(the last blocks might be shorter), denoting the $i$th blocks by $X_i$ and $Y_i$ for $i\in [0\dd m)$.
%\rknote{I think partition is much more appropriate for strings than decompose (here and throughout). }
For a \emph{sampling rate} $\rho:=\frac{b^2}{\phi n}$, the algorithm performs $\ceil{m\rho}$ iterations.
In each iteration, the algorithm chooses $i\in \fr{0}{m}$ uniformly at random and makes an oracle call to solve an instance $(X_{i},Y_{i})$ of \GEDa{\phi}{\beta}.
The algorithm returns YES if and only if all oracle calls return YES.\@

The total length of substrings involved in the oracle calls is
$\Oh(\rho m\cdot b)=\Oh(\rho n) = \Oh(\frac{b^2}{\phi})=\Oh(\frac{\phi n^2}{\alpha^2})$,
whereas the running time and number of oracle calls are 
$\Oh(\rho m)=\Oh(\frac{b}{\phi})=\Oh(\frac{n}{\alpha})$.

To prove the correctness of this reduction (assuming the oracle makes no errors), 
let $B:=\{i\in [0\dd m) : \ED(X_i,Y_i)>\phi\}$
and observe that $\ED(X,Y)\le |B| b + m\phi$.
Hence, if $\ED(X,Y)>\alpha$, then \[|B| > \tfrac{\alpha - m \phi}{b} \ge \tfrac{\alpha b-2n\phi}{b^2}
 \ge \tfrac{3\phi n - 2\phi n}{b^2}=\tfrac{1}{\rho}.\]
The probability that the algorithm returns YES is then at most
\[
  \left(1-\tfrac{|B|}{m}\right)^{\ceil{\rho m}}
  \le \exp\left(-\tfrac{|B|}{m}\cdot \ceil{\rho m}\right)
  \le \exp\left(-\rho|B|\right)\le \exp(-1).
\]
On the other hand, if $\ED(X,Y)\le \beta$, then \cref{fct:hered} implies that $\ED(X_{i},Y_{i})\le \beta$ for all $i\in \fr{0}{m}$. Consequently, all oracle calls return YES, and so does the entire algorithm.
\end{proof}

\subsection{Our Reduction}\label{sec:simpleReduction}

%In this section we prove the following theorem. 

\begin{theorem}\label{thm:our}
There exists a randomized reduction that, given a parameter $\phi\in \Zp$
and an instance of \GEDa satisfying $\frac1{10}\alpha\ge \phi \ge \beta \ge 1$,
solves the instance using $\Oh(\frac{n}{\alpha})$
non-adaptive calls to an oracle for \GEDa{\phi}{\beta}
involving substrings of total length $\Oh(\frac{\phi n \log n}{\alpha})$.
The reduction takes $\Oh(\frac{n}{\alpha})$ time,
does not access the input strings,
and errs with probability at most~$\frac{1}{e}$.
\end{theorem}

\paragraph*{The algorithm}
For every \emph{level} $p \in [0\dd \ceil{\log n}]$, partition $X,Y$ into $m_p:=\lceil\frac{n}{2^p}\rceil$ \emph{blocks} of length $2^p$ (the last blocks might be shorter), given by $ X_{p,i} = X\fr{i\cdot2^p}{\min(n,(i+1)2^p)}$ and $Y_{p,i} = Y\fr{i\cdot2^p}{\min(n,(i+1)2^p)}$ for $i\in \fr{0}{m_p}$.

Let $\rho:= \frac{10\phi}{\alpha}$. 
For each level $p\in [\ceil{\log \phi}\dd \floor{\log (\rho n)}]$, our algorithm performs $\ceil{\rho m_p}$ iterations.
In each iteration, the algorithm chooses $i\in \fr{0}{m_p}$ uniformly at random and calls an oracle to solve an instance $(X_{p,i},Y_{p,i})$ of the \GEDa{\phi}{\beta} problem. The algorithm returns YES if and only if all oracle calls (across all levels) return YES.\@

\paragraph*{Complexity Analysis}
The total length of substrings involved in the oracle calls is
\[\Oh\left(\sum_{p=\ceil{\log \phi}}^{\floor{\log(\rho n)}} 2^p\cdot \ceil{\rho m_p}\right)
=\Oh(\rho n \log n)= \Oh(\tfrac{\phi n\log n}{\alpha}).\]
The running time and the number of oracle calls are
\[\Oh\left(\sum_{p = \ceil{\log \phi}}^{\floor{\log (\rho n)}} \ceil{\rho m_p}\right)
= \Oh\left(\sum_{p=\ceil{\log \phi}}^{\floor{\log (\rho n)}} \tfrac{\rho n }{2^p}\right)  =
\Oh(\tfrac{\rho n}{\phi})=\Oh\left(\tfrac{n}{\alpha}\right).\]

\paragraph*{Correctness}
The core of the analysis is the following lemma,
which proves that an instance with large edit distance must contain,
across all the levels, many blocks of ``high cost''.
In the lemma, these blocks are denoted by $B_p$ for level $p$,
as illustrated in Figure~\ref{fig:blocks}. 

\begin{lemma}\label{lem:key}
  Consider a threshold $\tau\in \Zp$.
  For each level $p\in [0\dd \ceil{\log n}]$,
  let \[B_p:= \{i\in \fr{0}{m_p}: \ED(X_{p,i},Y_{p,i})>\tau\}.\]
  If $\ED(X,Y)> \tau$,
  then $\sum_{p=\ceil{\log \tau}}^{\ceil{\log n}}|B_p| > \frac{1}{2\tau}\ED(X,Y)$.
  \end{lemma}

  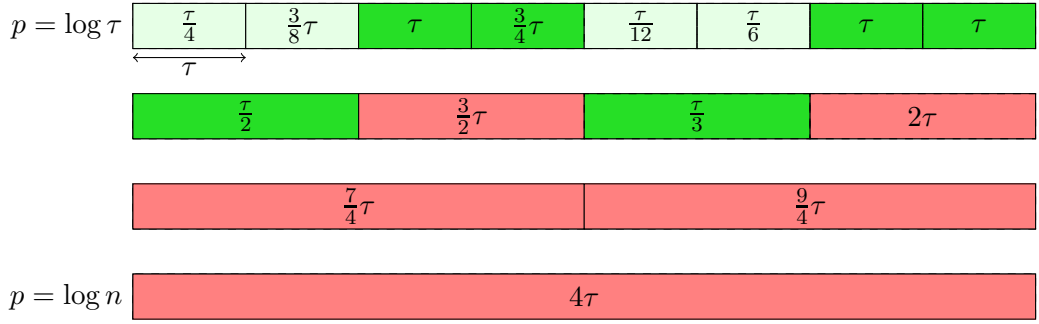
\begin{figure}[ht]
    \centering
\begin{tikzpicture}[yscale=1.2]
    \draw[] (-9,0) rectangle (3,0.5);
    \draw[dashed] (-9,0) rectangle (-7.5,0.5);
    \draw[<->] (-9,-0.1) -- (-7.5,-0.1);
    \node[] at (-8.25,-0.22) {$\tau$};

    \draw[dashed] (-7.5,0) rectangle (-6,0.5);
    \draw[dashed] (-6,0) rectangle (-4.5,0.5);
    \draw[dashed] (-4.5,0) rectangle (-3,0.5);
    \draw[fill=green!10] (-9,0) rectangle (-7.5,0.5);
    \draw[fill=green!10] (-7.5,0) rectangle (-6,0.5);
    \draw[fill=green!85!black!85] (-6,0) rectangle (-4.5,0.5);
    \draw[fill=green!85!black!85] (-4.5,0) rectangle (-3,0.5);
    \draw[fill=green!85!black!85] (0,0) rectangle (1.5,0.5);
    \draw[fill=green!10] (-3,0) rectangle (-1.5,0.5);
    \draw[fill=green!10] (-1.5,0) rectangle (0,0.5);
    \draw[fill=green!85!black!85] (1.5,0) rectangle (3,0.5);
    \draw[dashed] (-3,0) rectangle (-1.5,0.5);
    \draw[dashed] (-1.5,0) rectangle (0,0.5);
    \draw[dashed] (0,0) rectangle (1.5,0.5);
    \draw[dashed] (1.5,0) rectangle (3,0.5);
    
    \node[] at (-8.25,0.25) {$\frac \tau 4$};
    \node[] at (-6.75,0.25) {$\frac 3 8 \tau$};
    \node[] at (-5.25,0.25) {$\tau$};
    \node[] at (-3.75,0.25) {$\frac 3 4\tau$};
    \node[] at (-2.25,0.25) {$\frac \tau {12}$};
    \node[] at (-0.75,0.25) {$\frac \tau 6$};
    \node[] at (0.75,0.25) {$\tau$};
    \node[] at (2.25,0.25) {$\tau$};
    
    \draw (-9,0.25) node[left] {$p= \log \tau$};

    \draw[] (-9,-1) rectangle (3,-0.5);
    \draw[dashed] (-9,-1) rectangle (-6,-0.5);
    % \draw[<->] (-9,-1.2) -- (-6,-1.2);
    % \node[] at (-7.5,-1.35) {$2^{p}$};
    \draw[dashed] (-6,-1) rectangle (-3,-0.5);
    \draw[fill=red!50] (-6,-1) rectangle (-3,-0.5);
    \draw[fill=red!50] (0,-1) rectangle (3,-0.5);
    \draw[fill=green!85!black!85] (-9,-1) rectangle (-6,-0.5);
    \draw[fill=green!85!black!85] (-3,-1) rectangle (0,-0.5);
    \draw[dashed] (-3,-1) rectangle (0,-0.5);
    \draw[dashed] (0,-1) rectangle (3,-0.5);
    
    \node[] at (-7.5,-0.75) {$\frac \tau 2$};
    \node[] at (-4.5,-0.75) {$\frac 3 2\tau $};
    \node[] at (-1.5,-0.75) {$\frac \tau 3$};
    \node[] at (1.5,-0.75) {$2\tau$};

    \draw[] (-9,-2) rectangle (3,-1.5);
    \draw[dashed] (-9,-2) rectangle (-3,-1.5);
    % \draw[<->] (-9,-1.7) -- (-3,-1.7);
    % \node[] at (-7.5,-1.35) {$2^{p}$};
    \draw[dashed] (-3,-2) rectangle (3,-1.5);
    \draw[fill=red!50] (-9,-2) rectangle (-3,-1.5);
    \draw[fill=red!50] (-3,-2) rectangle (3,-1.5);
    % \draw[dashed] (-3,-1) rectangle (0,-0.5);
    % \draw[dashed] (0,-1) rectangle (3,-0.5);
    \node[] at (-6,-1.75) {$\frac 7 4\tau$};
    \node[] at (0,-1.75) {$\frac 9 4 \tau$};
    
    \draw[] (-9,-3) rectangle (3,-2.5);
    \draw[dashed] (-9,-3) rectangle (3,-2.5);
    % \draw[<->] (-9,-1.7) -- (-3,-1.7);
    % \node[] at (-7.5,-1.35) {$2^{p}$};
    % \draw[dashed] (-3,-2) rectangle (3,-1.5);
    \draw[fill=red!50] (-9,-3) rectangle (3,-2.5);
    % \draw[fill=red!50] (-3,-2) rectangle (3,-1.5);
    \node[] at (-3,-2.75) {$4\tau$};
    
    \draw (-9,-2.75) node[left] {$p= \log n$};
    
    \end{tikzpicture}
    \caption{
      The multi-level partitioning of $[0\dd n)$ into blocks of length $2^p$ at each level $p$.
      Inside each rectangle we write the edit distance
      between the corresponding pair of blocks in $X$ and in $Y$.
      % (e.g., $\ED(X_{p,i}$ and $Y_{p,i})$). 
      A red color represents that the block is in $B_{p}$ (high cost),
      and the two shades of green represent the remaining blocks. The dark green blocks color represents blocks in $\hat{B}_p \sm B_p$ (green blocks with a red `parent'); the crux of \cref{lem:key} is that the entire range can be decomposed into few such blocks. 
    }\label{fig:blocks}
  \end{figure}

\begin{proof}
We prove by induction on $p$ that \[\sum_{i\in B_p} \ED(X_{p,i},Y_{p,i})\le 2\tau\cdot \sum_{q=\ceil{\log \tau}}^p |B_q|.\]
The base case of $p<\ceil{\log \tau}$ holds trivially
because then $\ED(X_{p,i},Y_{p,i})\le \tau$ for all $i\in \fr{1}{m_p}$, and thus $B_{p}=\emptyset$.
        
For the inductive step, consider $p\ge\ceil{\log \tau}$.
        Observe that each $i\in \fr{0}{m_p}$ satisfies $X_{p,i}=X_{p-1,2i}\cdot X_{p-1,2i+1}$
        (if $X_{p-1,2n_{p}+1}$ is undefined, we set it to be empty) and, similarly,
        $Y_{p,i}=Y_{p-1,2i}\cdot Y_{p-1,2i+1}$.
        Consequently, we define $\hat{B}_{p-1} = \bigcup_{i\in B_p}\{2i,2i+1\}$, with $|\hat{B}_{p-1}| = 2|B_p|$, 
        and derive
\begin{align*}
  \sum_{i\in B_p} \ED(X_{p,i},Y_{p,i}) 
  \stackrel{\text{\cref{fct:subadd}}}{\le} & \sum_{j\in \hat{B}_{p-1}} \ED(X_{p-1,j},Y_{p-1,j})
  \\
  \stackrel{\text{\phantom{induction}}}{\le} & \sum_{j\in \hat{B}_{p-1}\sm B_{p-1}}  \ED(X_{p-1,j},Y_{p-1,j}) + \sum_{j\in B_{p-1}} \ED(X_{p-1,j},Y_{p-1,j})
  \\
   \stackrel{\text{induction}}{\le} & \tau\cdot|\hat{B}_{p-1} \sm B_{p-1}| + 2\tau\cdot\sum_{q=\ceil{\log\tau}}^{p-1} |B_q|
  \\  \stackrel{\text{\phantom{induction}}}{\le} & \tau\cdot|\hat{B}_{p-1}| + 2\tau\cdot\sum_{q=\ceil{\log\tau}}^{p-1} |B_q|\\
  \stackrel{\text{\phantom{induction}}}{=} & 2\tau\cdot\sum_{q=\ceil{\log \tau}}^{p} |B_q|.
\end{align*}
This completes the inductive proof.

The lemma follows by applying the inequality proved above and observing that $\ED(X,Y)>\tau$ implies $B_{\ceil{\log n}}=\{0\}$:
\[\ED(X,Y)
= \ED(X_{\ceil{\log n},0},Y_{\ceil{\log n},0})
\le 2\tau \!\!\sum_{p=\ceil{\log \tau}}^{\ceil{\log n}}\!\!|B_p|.\qedhere\]
\end{proof}

Let us proceed with the correctness analysis of our algorithm.
First, suppose that $\ED(X,Y)>\alpha$.
Using \cref{lem:key} with $\tau := \phi$
and the fact $\alpha \ge 10\phi > \tau$,
we then obtain
\[
  \sum_{p=\ceil{\log \phi}}^{\ceil{\log n}} |B_p| \ge \tfrac{\ED(X,Y)}{2\tau} > \tfrac{\alpha}{2\phi} = \tfrac{5}{\rho}.
\]
Using a naive bound
\[\sum_{p=\ceil{\log (\rho n)}}^{\ceil{\log n}}|B_p|
\le \sum_{p=\ceil{\log (\rho n)}}^{\ceil{\log n}}m_p
\le \sum_{p=\ceil{\log (\rho n)}}^{\ceil{\log n}}\tfrac{2n}{2^p}
\le \tfrac{4}{\rho},\]
we conclude that
$\sum_{p=\ceil{\log \phi}}^{\floor{\log (\rho n)}} |B_p| \ge \frac{1}{\rho}$.
For each level $p\in [\ceil{\log \phi}\dd \floor{\log (\rho n)}]$,
the probability that a single oracle (at a fixed iteration) call returns YES
is at most
$1-\frac{|B_p|}{m_p} \le \exp\big({-\frac{|B_p|}{m_p}}\big)$.
Across all levels $p\in [\ceil{\log \phi}\dd \floor{\log (\rho n)}]$,
the probability that all calls return YES is at most
\[
  \exp\left(-\sum_{p=\ceil{\log \phi}}^{\floor{\log(\rho n)}}\tfrac{|B_p|}{m_p}\cdot \ceil{\rho m_p}\right)
\le \exp\left(-\sum_{p=\ceil{\log \phi}}^{\floor{\log(\rho n)}}\rho|B_p|\right)
 \le 
\exp(-1).
\]
Thus, the algorithm returns YES with probability at most~$\frac{1}{e}$.

Now, suppose that $\ED(X,Y)\le \beta$.
Then, \cref{fct:hered} implies $\ED(X_{p,i},Y_{p,i})\le \beta$ for all $p\in [0\dd \ceil{\log n}]$ and $i\in \fr{0}{m_p}$.
Consequently, each oracle call returns YES, so our algorithm also returns YES.\@
This completes the proof of \cref{thm:our}.

\subsection{Corollaries (by Plugging Known Algorithms)}\label{sec:simpleAlgCor}

\begin{corollary}\label{cor:our}
There exists a non-adaptive randomized algorithm that, given parameters $\eps,\delta \in \R_+$ and an instance
of \GEDa satisfying $\alpha \ge \floor{\fAN(\eps)\beta}$,
solves the instance in time
$\Oh(\frac{1+\beta}{1+\alpha} \fAN(\eps)\gAN(\eps)\cdot n^{1+\eps}\log^2 n \cdot \log\frac{1}{\delta})$,
using $\Oh(\frac{1+\beta}{1+\alpha} \fAN(\eps) \cdot n\log n \cdot \log\frac{1}{\delta})$ queries to the input strings, 
and with error probability at most $\delta$,
where $\fAN$ and $\gAN$ are the functions of \cref{thm:an}.
\end{corollary}

\begin{proof}
If $\beta = 0$, then we use the algorithm of \cref{fct:hm},
which takes time $\Oh(\frac{n}{1+\alpha})$. 
If $\floor{\fAN(\eps)\beta}\le \alpha < 10\fAN(\eps)\beta$, we use the algorithm of \cref{cor:an},
which takes $\Oh(\gAN(\eps)n^{1+\eps})$ time and $\Oh(n)$ queries. 
In the remaining case of $0 < 10 \fAN(\eps)\beta \le \alpha$,
we use the reduction of \cref{thm:our} with $\phi = \floor{ \fAN(\eps)\beta }$
and the oracle implemented using \cref{cor:an}. 
The oracle is randomized, so we need to set its error probability to $\Theta(\frac{\alpha}{n})$
so that all oracle calls are correct with large constant probability.
An oracle call involving a pair of strings of length $m$ takes time
$\Oh(\gAN(\eps) m^{1+\eps} \log n)
= \Oh(\gAN(\eps) m \cdot n^{\eps} \log n)$,
and the total length of all strings involved in the oracle calls is
$\Oh(\frac{\beta}{\alpha}\cdot \fAN(\eps) n\log n)$;
therefore, the total running time is
$\Oh(\frac{\beta}{\alpha}\cdot \fAN(\eps)\gAN(\eps)\cdot n^{1+\eps}\log^2 n)$.
This completes the algorithm's description for $\delta > \frac{1}{e}$.
For general $\delta>0$,
we amplify the success probability by taking the majority answer among
$\Oh(\log\frac{1}{\delta})$ independent repetitions of the entire algorithm.
\end{proof}

Next, we observe that the running time can be expressed as $\hat{\Oh}(\frac{1+\beta}{1+\alpha}\cdot n)$
as long as $\frac{\alpha}{\beta}=\omega(1)$.

\begin{corollary}\label{cor:simple}
Let $s:\Zz \to \Zz$ be a function such that $\lim_{x\to \infty}\frac{s(x)}{x}=\infty$.
There exists a randomized algorithm that solves any instance of \GEDa with $\alpha \ge s(\beta)$
in time $\hat{\Oh}(\frac{1+\beta}{1+\alpha}\cdot n)$ correctly with high probability.
\end{corollary}
\begin{proof}
Observe that there exists a function $\eps: \Zz \to \R_+$ such that $\eps(n)=o(1)$, $\fAN(\eps(n)) \le \log n$, $\gAN(\eps(n))\le \log n$, and $\fAN(\eps(n))\le \frac{s(x)}{x}$ for all $n\in \Zz$ and $x> \log n$.
If $\alpha \ge \fAN(\eps(n))\beta$,
we use \cref{cor:our} with $\eps = \eps(n)$,
which takes
$\Ohtilde(\frac{1+\beta}{1+\alpha}\fAN(\eps(n))\gAN(\eps(n))n^{1+\eps(n)})=\hat{\Oh}(\frac{1+\beta}{1+\alpha}n)$ time.
If $\alpha < \fAN(\eps(n))\beta$ and $\beta \le \log n$,
we use \cref{thm:lv}, which takes
$\Oh(n+\beta^2)=\Oh(n)=\Oh(\frac{1+\alpha}{1+\beta}\cdot \frac{1+\beta}{1+\alpha}\cdot n)
=\Oh(\fAN(\eps(n))\cdot \frac{1+\beta}{1+\alpha}\cdot n)=\hat{\Oh}(\frac{1+\beta}{1+\alpha}\cdot n)$ time.
In the remaining case of $\alpha < \fAN(\eps(n))\beta$ and $\beta > \log n$,
we have $\alpha < \fAN(\eps(n))\beta \le s(\beta)$,
which contradicts our assumption $\alpha \ge s(\beta)$.
\end{proof}

\section{Improved Query Complexity}\label{sec:qc}
In this section, we improve the query complexity of the algorithm described in~\cref{sec:simpleAlg}, solving the \GEDa{\alpha}{\beta} with query complexity $\Ohtilde\left(\frac{n\sqrt \beta}{\alpha}\right)$ provided that $\alpha \gg \beta$.

\subsection{Overview}
Recall that \cref{thm:our} provides a randomized reduction from the \GEDa problem to the \GEDa{\phi}{\beta} problem.
We used $\ED(X_{i,p},Y_{i,p})\le \ED(X,Y)$ to justify the correctness for YES instances: The input can be safely rejected as soon as we discover that  $\ED(X_{i,p},Y_{i,p})> \beta$ holds for some level $p$ and index $i\in [0\dd m_p)$. 
If we were guaranteed that $\ED(X_{i,p},Y_{i,p})\le \psi$ holds with good probability (over random $i\in [0\dd m_p)$) for some $\psi < \beta$, then we could use an oracle for the \GEDa{\phi}{\psi} problem
instead of the \GEDa{\phi}{\beta} problem, i.e., our reduction would produce instances of the \GED problem with a larger gap.
Unfortunately, this is not the case in general.
In particular, if $X$ consists of distinct characters and $Y$ is obtained by moving the last $s\le \frac12n$ characters of $X$ to the front, then $\ED(X_{i,p},Y_{i,p})=\ED(X,Y)=2s$ holds for all levels $p\ge \log (2s)$ and indices $i\in [0\dd m_p)$.
Nevertheless, in this example, the optimal alignment between $X_{i,p}$ and $Y_{i,p}$ is very simple:
up to the shift by $s$ characters (which effectively removes the last $s$ characters of $X_{i,p}$ and the first $s$ characters of $Y_{i,p}$), the two blocks are perfectly aligned.
In general, for a fixed alignment $\mathcal{A}$ of $X$ and $Y$, the induced alignment of $X_{i,p}$ and $Y_{i,p}$
performs the edits that $\A$ would make on $X_{i,p}$ and $Y_{i,p}$,
and the only effect of edits that $\A$ makes outside these blocks is that some leading and trailing characters of  $X_{i,p}$ and $Y_{i,p}$ need to be deleted (because $\A$ aligns them with characters outside the considered blocks). 
Thus, in a YES-instance, for a random $i\in [0\dd m_p)$, we always see up to $\beta$ edits between $X_{i,p}$ and $Y_{i,p}$, but in expectation only $\frac{\beta}{m_p}$ of these edits cannot be attributed to a \emph{shift} between $X_{i,p}$ and $Y_{i,p}$. This motivates the following notion.

\begin{definition}
For two strings $X,Y\in \Sigma^*$ and a threshold $\beta\in \Zz$,
define the $\beta$-\emph{shifted edit distance} $\ED_{\beta}(X,Y)$ as
\[
   \min\left(\bigcup_{\Delta=0}^{\min(|X|,|Y|,\beta)}\big\{\ED(X[\Delta\dd |X|),Y[0\dd |Y|-\Delta)), \ED(X[0\dd |X|-\Delta), Y[\Delta\dd |Y|))\big\}\right).\]
Note that $\ED_{\beta}(X,Y)\le \ED(X,Y)\le \ED_{\beta}(X,Y)+2\beta$ holds for every $\beta\in \Zz$.
\end{definition}

As argued above, the YES-instances of \GEDa satisfy $\Exp_i[\ED_{\beta}(X_{p,i},Y_{p,i})] \le \frac{\beta}{m_p}$.
Given that we sample blocks with rate $\rho$, we expect to see $\Oh(1)$ blocks with $\ED_{\beta}(X_{p,i},Y_{p,i}) > \psi$ 
if we appropriately set $\psi = \Thetatilde(\rho \beta)$. Moreover, this statement is also true with high probability. 
Furthermore, the argument in the proof of \cref{thm:our} can be strengthened to prove that, in a NO-instance, with high probability, 
we see $\Omega(1)$ blocks with $\ED(X_{p,i},Y_{p,i}) > 3\phi$.
Thus, instead of using an oracle for the \GEDa{\phi}{\beta} problem,
we can use an oracle for the \SEDa{\phi}{\beta}{\psi} problem defined as follows.

\begin{problem}[\SEDa]
    Given strings $X,Y\in \Sigma^n$ and integer thresholds $\alpha\ge \beta \ge \gamma \ge 0$, return
    YES if $\ED_{\beta}(X,Y)\le  \gamma$,
    NO if $\ED(X,Y) > 3\alpha$,
    %  \egnote{Why not rejecting if $\ED(X,Y)>\alpha$ and require $\alpha>2\beta+\gamma$?}
    % \tknote{This parametrization is just to keep the formulas (the running times, assumptions, etc.) reasonably simple. In the natural formulation, we would need to replace $\alpha$ with $\alpha-2\beta$ everywhere.}\egnote{In that case maybe we should leave it, although I'm doubtful if anyone remembers the constant $3$ along the proofs.}
    and an arbitrary answer otherwise.   
\end{problem}

The idea to separate the shift from the ``local'' edits originates from~\cite{BEKMRRS03},
but they were only able to solve the \GEDa problem for $\alpha = \Omega(n)$.
Combining their insight into our reduction of \cref{thm:our}, we can handle a much wider range of parameters.

Similarly to~\cite{BEKMRRS03}, our algorithm is recursive in nature, with $\beta$ decreased in each level (until it reaches 0). 
There is a key difference, though: 
They reduce the \SED problem to a more general problem,
which becomes even more complicated in subsequent recursion levels. 
Our new insight is that, surprisingly, the \SED problem
can be reduced back to (multiple instances of) \GED.\@
This yields an algorithm with a much cleaner structure,
and furthermore improves the query complexity
because all these instances operate on relatively few different input strings, 
which can be easily exploited due to the non-adaptive nature of our approach.
In fact, this query complexity is optimal (up to $\log^{\Oh(1)}(n)$ terms) for non-adaptive algorithms,
as indicated by our lower bound, which generalizes the one in~\cite{BEKMRRS03}.

In this section, we present a solution that achieves this optimal query complexity but does not significantly improve the running time compared to \cref{sec:simpleAlg}.
The latter issue is addressed in \cref{sec:faster}, 
where we exploit dependencies between the \SED instances produced throughout the recursive calls. 
Specifically, we provide a more efficient implementation for the batched version of the \SEDa{\alpha}{\beta}{0} problem
arising at the lowest level of our recursion. Moreover, we carefully adjust the parameters 
at the three lowest levels of recursion so that they produce batches with desirable properties.
Up to logarithmic factors, our time bound matches that of~\cite{BEKMRRS03}, 
but the latter is valid only for $\alpha = \Omega(n)$. 

% The key idea behind our improved algorithm is that for any pair of strings $(X,Y)$ satisfying $\ED(X,Y)$ for most pairs of blocks it is not only the case that $\ED(X_i,Y_i)<\beta$ but rather $\ED_{\beta}(X_i,Y_i)\le \gamma \ll \beta$. Hence instead of testing small edit distance between sampled pairs we can test for their shifted edit distance.

\subsection{From \GED to \SED}\label{sec:gs}

Below, we reduce the \GEDa problem to the \SEDa{\phi}{\beta}{\psi} problem,
where $\phi\ge \beta$ can be adjusted and  $\psi$ is set to $\Ohtilde(\frac{\beta\phi}{\alpha})$.
Our immediate application in \cref{sec:baseline} uses $\phi=\beta$, but subsequent speedups in \cref{sec:faster} sometimes require $\phi \gg \beta$.
\begin{lemma}\label{lem:gs}
There exists a randomized reduction that, given a parameter $\phi\in \Zp$
and an instance of \GEDa satisfying $\phi \ge \beta \ge \psi:=\floor{\frac{112\beta\phi \ceil{\log n}}{\alpha}}$, solves the instance using $\Oh(\frac{n}{\alpha})$ non-adaptive calls to an oracle for \SEDa{\phi}{\beta}{\psi} involving substrings of total length $\Oh(\frac{\phi n\log n}{\alpha})$.
The reduction costs $\Oh(\frac{n}{\alpha})$ time,
does not access the input strings,
and errs with probability at most~$\frac{1}{e}$.
\end{lemma}
\begin{proof}
Let $\rho = \frac{84\phi}{\alpha}$ and $\tau=3\phi$.
For each level $p\in [\ceil{\log \tau}\dd \floor{\log (\rho n)}]$, our algorithm performs  $\ceil{\rho m_p}$ iterations.
In each iteration, the algorithm chooses $i\in \fr{0}{m_p}$ uniformly at random and solves an instance $(X_{p,i},Y_{p,i})$ of the \SEDa{\phi}{\beta}{\psi} problem. 
Finally, the algorithm returns YES if the number $\hat{b}$ of oracle calls with NO answers satisfies $\hat{b}\le 5$;
if $\hat{b}\ge 6$, the algorithm returns NO.\@

Let us first analyze the complexity of the algorithm.
The total length of all strings involved in the oracle calls is \[\Oh\left(\sum_{p=\ceil{\log \tau}}^{\floor{\log(\rho n)}} 2^p\cdot  \ceil{\rho m_p}\right)
=\Oh(\rho n \log n)= \Oh(\tfrac{\phi n\log n}{\alpha}).\] 
The running time and the number of oracle calls are
\[\Oh\left(\sum_{p = \ceil{\log \tau}}^{\floor{\log (\rho n)}} \ceil{\rho m_p}\right)
= \Oh\left(\sum_{p=\ceil{\log \tau}}^{\floor{\log (\rho n)}} \tfrac{\rho n}{2^p}\right)  =
\Oh\left(\tfrac{\rho n}{\tau}\right)=\Oh\left(\tfrac{n}{\alpha}\right).\]

Let us now proceed with the algorithm correctness.
If $\ED(X,Y)>\alpha$, then we use \cref{lem:key}.
Due to $\alpha \ge 112\phi\ceil{\log n} > 3\phi= \tau$, we have $\sum_{p=\ceil{\log\tau}}^{\ceil{\log n}} |B_p| > \frac{\alpha}{2\tau}=\frac{14}{\rho}$.
At the same time, $\sum_{p=\ceil{\log (\rho n)}}^{\ceil{\log n}}|B_p| \le \sum_{p=\ceil{\log n}}^{\infty}\frac{2n}{2^p} \le \frac{4}{\rho}$, so
$\sum_{p=\ceil{\log\tau}}^{\floor{\log (\rho n)}} |B_p| \ge \frac{10}{\rho}$.
For a fixed iteration at level $p\in [\ceil{\log\tau}\dd \ceil{\log n}]$, the probability that the oracle call returns NO is at least $\frac{|B_p|}{m_p}$.
Across all iterations and all levels $p \in [\ceil{\log\tau}\dd \floor{\log (\rho n)}]$, the expected number of NO answers is therefore 
\[\Exp\left[\hat{b}\right] \ge \sum_{p=\ceil{\log\tau}}^{\floor{\log(\rho n)}} \frac{|B_p|\ceil{\rho m_p}}{m_p}
\ge \rho\cdot \sum_{p=\ceil{\log\tau}}^{\floor{\log (\rho n)}}|B_p|\ge 10.\]
By the Chernoff bound, we thus have
\[\Pr\left[\hat{b} \le 5\right] = \Pr\left[\hat{b} \le (1-\tfrac12)\cdot 10\right] \le \exp\left(-\tfrac{(\frac12)^2 \cdot 10}{2}\right)  = \exp(-\tfrac{5}{4}) < \exp(-1).\]

Finally, consider the case of $\ED(X,Y)\le \beta$.
For every level $p\in [0\dd \ceil{\log n}]$, we define a set $G_p= \{i\in \fr{1}{m_p} : \ED_{\beta}(X_{p,i},Y_{p,i})>\psi\}$ corresponding to oracle calls that may return NO.\@
\begin{claim}\label{clm:gp}
We have $\sum_{p=1}^{\lceil \log n\rceil} |G_p| \le \frac{2\beta \ceil{\log n}}{\psi+1}\le \frac{3}{2\rho}$.
\end{claim}
\begin{proof}
For each level $p\in [1\dd \ceil{\log n}]$, let us consider a partition $Y=\bigodot_{i\in [0\dd m_p)} Y'_{p,i}$ such that $\ED(X,Y)=\sum_{i\in [0\dd m_p)}\ED(X_{p,i},Y'_{p,i})$.
We claim that $\ED_{\beta}(X_{p,i},Y_{p,i})\le 2\ED(X_{p,i},Y'_{p,i})$.
If $|Y'_{p,i}|\le \beta$, then $\ED_{\beta}(X_{p,i},Y_{p,i})\le \max(0,|X_{p,i}|-\beta)
\le \ED(X_{p,i},Y'_{p,i})$ and the claim holds trivially.
Thus, we assume $|Y'_{p,i}|>\beta$ and consider two cases. 
\begin{figure}[b!]
    \centering
\begin{tikzpicture}
    % \draw[] (-9,0) rectangle (3,0.5);
    \draw[] (-9,0) rectangle (-6,0.5);
    % \draw[<->] (-9,0.6) -- (-7.5,0.6);
    
    % \draw[dashed] (-7.5,0) rectangle (-6,0.5);
    \draw[] (-6,0) rectangle (-3,0.5);
    % \draw[dashed] (-4.5,0) rectangle (-3,0.5);
    % \draw[fill=red] (-6,0) rectangle (-4.5,0.5);
    % \draw[fill=red] (-4.5,0) rectangle (-3,0.5);
    \draw[] (-3,0) rectangle (0,0.5);
    % \draw[dashed] (-1.5,0) rectangle (0,0.5);
    \draw[] (0,0) rectangle (3,0.5);
    % \draw[dashed] (1.5,0) rectangle (3,0.5);

    % \draw[] (-9,-1) rectangle (3,-0.5);
    \draw[] (-9,-1.5) rectangle (-5,-1);
    % \draw[<->] (-9,-1.2) -- (-6,-1.2);
   
    \draw[] (-5,-1.5) rectangle (-2.75,-1);
    % \draw[fill=red] (-6,-1) rectangle (-3,-0.5);
    \draw[] (-2.75,-1.5) rectangle (0.3,-1);
    \draw[] (0.3,-1.5) rectangle (3,-1);
    
    \draw[dashed, ->] (-6,0) -- (-5,-1);
    \draw[dashed, ->] (-6,0) -- (-6,-1);
    
    \draw[ <->] (-6,-0.9) -- (-5,-0.9);
    
    \node[] at (-5.5,-0.75) {$\Delta$};
    
    \node[] at (-10,-1.25) {$Y$};
    
    \node[] at (-10,0.25) {$X$};
    
    \node[] at (-7.5,0.25) {$X_{p,1}$};
    \node[] at (-4.5,0.25) {$X_{p,2}$};
    \node[] at (-1.5,0.25) {$X_{p,3}$};
    \node[] at (1.5,0.25) {$X_{p,4}$};

    \node[] at (-7,-1.25) {$Y'_{p,1}$};
    \node[] at (-3.75,-1.25) {$Y'_{p,2}$};
    \node[] at (-1,-1.25) {$Y'_{p,3}$};
    \node[] at (1.5,-1.25) {$Y'_{p,4}$};

    \end{tikzpicture}
    \caption{The partitions $X=\bigodot_{i\in [0\dd m_p)} X_{p,i}$ and $Y=\bigodot_{i\in [0\dd m_p)} Y'_{p,i}$. }\label{fig:good}
\end{figure}
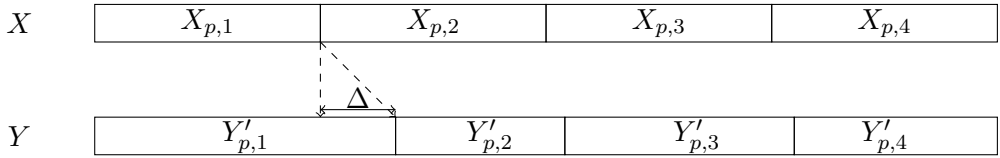

First, suppose that $Y'_{p,i}$ starts at position $i\cdot 2^p+\Delta$ for $\Delta \ge 0$.
We then have 
$\Delta \le \ED(X[0\dd i\cdot 2^p),Y[0\dd i\cdot 2^p+\Delta)) = \sum_{j=0}^{i-1}\ED(X_{p,j},Y'_{p,j})\le \ED(X,Y)\le \beta$ and thus:
\begin{align*}
\ED_{\beta}(X_{p,i},Y_{p,i}) &\le \ED(X_{p,i}[0\dd |X_{p,i}|-\Delta), Y_{p,i}[\Delta\dd |Y_{p,i}|))\\
&\le \ED(X_{p,i}[0\dd |X_{p,i}|-\Delta),Y'_{p,i}[0\dd |Y'_{p,i}|-\Delta)) + \big||Y_{p,i}|-|Y'_{p,i}|\big|\\
        &\le \ED(X_{p,i}, Y'_{p,i}) + \big||X_{p,i}|-|Y'_{p,i}|\big|\\
        &\le 2\ED(X_{p,i},Y_{p,i}).
\end{align*}
Similarly, if $Y'_{p,i}$ starts at position $i\cdot 2^p - \Delta$ for some $\Delta \ge 0$, then $\Delta \le \sum_{j=0}^{i-1}\ED(X_{p,j},Y'_{p,j})\le \ED(X,Y) \le \beta$
and
\begin{align*}
\ED_{\beta}(X_{p,i},Y_{p,i}) &\le \ED(X_{p,i}[\Delta\dd |X_{p,i}|), Y_{p,i}[0\dd |Y_{p,i}|-\Delta))\\
&\le \ED(X_{p,i}[\Delta \dd |X_{p,i}|),Y'_{p,i}[\Delta \dd |Y'_{p,i}|)) + \big||Y_{p,i}|-|Y'_{p,i}|\big|\\
&\le \ED(X_{p,i}, Y'_{p,i}) + \big||X_{p,i}|-|Y'_{p,i}|\big|\\
&\le 2\ED(X_{p,i},Y'_{p,i})    
\end{align*}
Thus, $\sum_{p=1}^{\ceil{\log n}} \sum_{i=0}^{m_p-1} \ED_\beta(X_{p,i},Y_{p,i}) \le 2\beta\ceil{\log n}$
and at most $\frac{2\beta\ceil{\log n}}{\psi+1}$ terms exceed $\psi$.
\end{proof}

For a fixed iteration at level $p\in [\ceil{\log\tau}\dd \floor{\log (\rho n)}]$, the probability that the oracle call returns NO
is at most $\frac{|G_p|}{m_p}$.
Across all iterations and levels $p\in [\ceil{\log\tau}\dd \floor{\log (\rho n)}]$, the expected number of NO answers is
therefore \[\Exp\left[\hat{b}\right]\le \sum_{p=\ceil{\log\tau}}^{\floor{\log (\rho n)}} \frac{|G_p|\ceil{\rho m_p}}{m_p}
\le 2\rho \sum_{p=\ceil{\log\tau}}^{\floor{\log (\rho n)}} |G_p| \le 3.\]
By the Chernoff bound, we thus have
\[\Pr\left[\hat{b} \ge 6\right] = \Pr\left[\hat{b} \ge (1+1)\cdot 3\right] \le \exp(-\tfrac{1^2\cdot 3}{2+1}) =\tfrac1e.\qedhere\]
\end{proof}

\subsection{From \SED to \GED}\label{sec:sg}
\begin{lemma}\label{lem:sg}
There exists a deterministic reduction that, given instance of \SEDa
satisfying $\alpha\ge 3\gamma$, solves the instance
using $\Oh\big(\frac{1+\beta}{1+\gamma}\big)$ non-adaptive calls to an oracle for \GEDa{\alpha}{3\gamma} involving $\Oh\big(\frac{\sqrt{1+\beta}}{\sqrt{1+\gamma}}\big)$ \emph{distinct} substrings. each of length at most~$n$.
The reduction takes $\Oh(\frac{1+\beta}{1+\gamma})$ time
and does not access the input strings. 
\end{lemma}
\begin{proof}
Let $\xi\in [\gamma\dd \beta]$ be a parameter set later and $n'=n-\beta$. 
We use the oracle to solve several instances of the \GEDa{\alpha}{3\gamma} problem
and return YES if and only if at least one oracle call returned YES.
These instances are $(X[x\dd x+n'),Y[y\dd y+n'))$
for all $x \in [0\dd \beta]$ such that $x\equiv 0 \pmod{1+\xi}$ or $x\equiv \beta \pmod{1+\xi}$,
all $y\in [0\dd \xi]$ such that $y\equiv 0 \pmod{1+\gamma}$,
and all $y\in [\beta-\xi\dd \beta]$ such that  $y\equiv \beta \pmod{1+\gamma}$.
The number of oracle calls is at most $4\big\lceil{\frac{1+\beta}{1+\xi}}\big\rceil\big\lceil{\frac{1+\xi}{1+\gamma}\big\rceil}\le 16\cdot \frac{1+\beta}{1+\gamma}$, 
but the number of distinct strings involved in these calls is at most $2\big\lceil{\frac{1+\beta}{1+\xi}}\big\rceil+2\big\lceil{\frac{1+\xi}{1+\gamma}}\big\rceil$, which is $\Oh\big(\frac{\sqrt{1+\beta}}{\sqrt{1+\gamma}}\big)$
if we set $1+\xi = \floor{\sqrt{(1+\beta)(1+\gamma)}}$.

Suppose that $\ED_{\beta}(X,Y)\le \gamma$.
First, consider the case when $\ED_{\beta}(X,Y)=\ED(X[\Delta\dd n),\allowbreak Y[0\dd n-\Delta))$ for some $\Delta \in [0\dd \beta)$.
Let us choose the smallest $x\in [\Delta\dd \beta]$ with $x\equiv \beta \pmod{1+\xi}$
and the largest $y\in [0\dd x-\Delta]$ with $y\equiv 0 \pmod{1+\gamma}$.
Observe that \begin{align*}\ED(X[x\dd x+n'),Y[y\dd y+n')) &\le 2\gamma + \ED(X[x\dd x+n'),Y[x-\Delta \dd x-\Delta+n'))\\
&\le 2\gamma+ \ED(X[\Delta\dd n),Y[0\dd n-\Delta))\le 3\gamma.\end{align*}
Hence, the oracle call for $(x,y)$ must return YES.\@

Similarly, let us consider the case when $\ED_{\beta}(X,Y)=\ED(X[0\dd n-\Delta),\allowbreak Y[\Delta \dd n))$ for some $\Delta \in [0\dd \beta)$.
Let us choose the largest $x\in [0\dd \beta-\Delta]$ with $x\equiv 0 \pmod{1+\xi}$
and the smallest $y\in [x+\Delta \dd \beta]$ with $y\equiv \beta \pmod{1+\gamma}$.
Observe that \begin{align*}\ED(X[x\dd x+n'),Y[y\dd y+n')) & \le 2\gamma + \ED(X[x\dd x+n'),Y[x+\Delta \dd x+\Delta+n'))\\
&\le 2\gamma+ \ED(X[0\dd n-\Delta),Y[\Delta \dd n))\le 3\gamma\end{align*}
Hence, the oracle call for $(x,y)$ must return YES.\@

Next, suppose that some oracle call for $(x,y)$ returned YES.\@
This implies $\ED(X[x\dd x+n'),\allowbreak Y[y\dd {y+n'}))\le \alpha$ for some $x,y\in [0\dd \beta]$.
At the same time, we have \begin{align*}
\ED(X[0\dd x),Y[0\dd y))&\le \max(x,y)\le \beta,\qquad\text{ and}\\
\ED(X[x+n'\dd n),Y[y+n'\dd n))&\le \max(\beta-x,\beta-y) \le \beta.\end{align*}
Hence, $\ED(X,Y)\le \alpha+2\beta\le 3\alpha$ holds as claimed.
\end{proof}

\subsection{Baseline Implementation}\label{sec:baseline}

\begin{proposition}\label{prp:baseline}
There exists a non-adaptive algorithm that, given $h\in \Zz$, $\delta\in \R_+$,
and an instance of the \GEDa problem,
satisfying  $\beta < (336\ceil{\log n})^{\frac{-h}2} \alpha^{\frac{h}{h+1}}$, solves the instance 
in $\Oh\big(\frac{1+\beta}{1+\alpha}\cdot n\log^{2h} n \cdot \log\frac{1}{\delta}\cdot 2^{\Oh(h)}\big)$ time, using $\Oh\big(\frac{\sqrt{1+\beta}}{1+\alpha}\cdot n\log^{2h}n\cdot \log\frac{1}{\delta}\cdot 2^{\Oh(h)}\big)$ queries, and with error probability at most $\delta$.
\end{proposition}
\begin{proposition}\label{prp:baseline2}
There exists a non-adaptive algorithm that, given $h\in \Zz$, $\delta\in \R_+$, and an instance of the \SEDa problem
satisfying $\gamma < \frac13(336\ceil{\log n})^{\frac{-h}2} \alpha^{\frac{h}{h+1}}$, solves the instance in $\Oh\big(\frac{1+\beta}{1+\alpha}\cdot n\log^{2h} n \cdot \log\frac{n}{\delta}\cdot 2^{\Oh(h)}\big)$ time,
using $\Oh\big(\frac{\sqrt{1+\beta}}{1+\alpha}\cdot n\log^{2h}n \cdot \log\frac{n}{\delta}\cdot 2^{\Oh(h)}\big)$ queries, and with error probability at most $\delta$.
\end{proposition}
\begin{proof}[Proof of \cref{prp:baseline,prp:baseline2}]
As for \GEDa, let us assume that $\delta > \frac{1}{e}$; in general, we amplify the success probability by repeating the algorithm $\Oh(\log\frac{1}{\delta})$ times.
If $\beta=0$ (and, in particular, $h=0$), we simply use \cref{fct:hm}.
Otherwise, we apply \cref{lem:gs} with $\phi = \beta$ using our \SEDa{\phi}{\beta}{\psi} algorithm (with parameters $h-1$ and $\Theta(\frac{1}{n})$) as the oracle.
This is valid because
\[\psi \le \tfrac{112 \beta^2 \ceil{\log n}}{\alpha} < \tfrac{112 \beta^2\ceil{\log n}}{\beta^{\frac{h+1}{h}}\cdot (336 \ceil{\log n})^{\frac{h+1}{2}}}  =
\tfrac13\cdot \left(336 \ceil{\log n}\right)^{\frac{1-h}{2}}\cdot \phi^{\frac{h-1}{h}} \le \tfrac13\cdot \phi < \beta.\]
The running time is \[
  \Oh\left(\tfrac{\phi \log n}{\alpha} \cdot \tfrac{\beta}{\phi}\cdot n\log^{2h-2}n\cdot \log n \cdot 2^{\Oh(h-1)}\right)
= \Oh\left(\tfrac{\beta}{\alpha}\cdot n\log^{2h}n \cdot 2^{\Oh(h)}\right),\]
whereas the query complexity is \[
  \Oh\left(\tfrac{\phi \log n}{\alpha} \cdot \tfrac{\sqrt{\beta}}{\phi}\cdot n\log^{2h-2}n \cdot \log n \cdot 2^{\Oh(h-1)}\right)
= \Oh\left(\tfrac{\sqrt{\beta}}{\alpha}\cdot n\log^{2h}n \cdot 2^{\Oh(h)}\right).\]

As for the \SEDa problem, we apply \cref{lem:sg} using our \GEDa{\alpha}{3\gamma} algorithm (with parameters $h$ and $\Theta(\frac{\delta}{n})$) as the oracle.
This is valid since $3\gamma < (336\ceil{\log n})^{\frac{-h}2} \alpha^{\frac{h}{h+1}} \le \alpha$.
The running time is  \[
\Oh\left(\tfrac{1+\beta}{1+\gamma} \cdot \tfrac{1+3\gamma}{1+\alpha}\cdot n\log^{2h}n \cdot \log\tfrac{n}{\delta} \cdot 2^{\Oh(h)}\right)\\
= \Oh\left(\tfrac{1+\beta}{1+\alpha}\cdot n\log^{2h}n \cdot \log\tfrac{n}{\delta} \cdot 2^{\Oh(h)}\right),\]
whereas the query complexity is 
\[
\Oh\left(\tfrac{\sqrt{1+\beta}}{\sqrt{1+\gamma}} \cdot \tfrac{\sqrt{1+3\gamma}}{1+\alpha}\cdot n\log^{2h}n \cdot \log\tfrac{n}{\delta} \cdot 2^{\Oh(h)}\right)\\
= \Oh\left(\tfrac{\sqrt{1+\beta}}{1+\alpha}\cdot n\log^{2h}n\cdot  \log\tfrac{n}{\delta} \cdot 2^{\Oh(h)}\right).\qedhere
\]
\end{proof}

\section{Faster Implementation}\label{sec:faster}
In this section, we improve the running time while preserving the query complexity behind \cref{prp:baseline}. 
The main trick is to consider a \emph{batched} version on the \GEDa and \SEDa problems:
Instances $(X_1,Y_1),\ldots,(X_q,Y_q)$ form a \emph{batch}
if $X_1=\cdots = X_q$.

\subsection{\SED for $h=0$}
\begin{lemma}\label{lem:s1}
There exists a non-adaptive algorithm that, given a parameter $\delta\in \R_+$,
and a batch of $q$ instances of \SEDa{\alpha}{\beta}{0},
solves the instances in $\Oh\big(\frac{\sqrt{q(q+\beta)}}{1+\alpha}\cdot n\log \frac{n}{\delta}\big)$ time with each answer correct with probability at least $1-\delta$. Moreover, at most $\Oh\big(\frac{1+\beta}{1+\alpha}\cdot n \log \frac{n}{\delta})$  characters of the common string $X$ are accessed.
\end{lemma}
\begin{proof}
    We simulate the algorithm in the proof of \cref{lem:sg}, setting $1+\xi = \big\lceil{\frac{\sqrt{q+\beta}}{\sqrt{q}}}\big\rceil$.
    For each instance, this yields $\Oh(1+\beta)$ oracle calls asking to solve the \GEDa{\alpha}{0} problem for $(X',Y')$
    with $|X'|=|Y'|=n'\le n$. The set of pairs $(X',Y')$ involved in these calls 
    can be expressed as $\mathcal{X}\times \mathcal{Y}$, where $|\mathcal{X}|=\Oh\big(\frac{1+\beta}{1+\xi}\big)=\Oh\big(\min\big(1+\beta,\sqrt{q(q+\beta)}\big)\big)$ and $|\mathcal{Y}|=\Oh(1+\xi)=\Oh\big(\frac{\sqrt{q+\beta}}{\sqrt{q}}\big)$.
    Moreover, since our algorithm is non-adaptive, the set $\mathcal{X}$ is the same for all $q$ instances.

    Recall that the reduction of \cref{lem:sg} returns YES if and only if at least one of the oracle calls returns YES.\@ To simulate implementing the calls using the algorithm of \cref{fct:hm}, we construct a random sample $S\sub [0\dd n')$ of size  $\Theta\big(\frac{n' \log\frac{n}{\delta}}{1+\alpha}\big)$.
    We build a set $\mathcal{X}_S:=\{X'[S] : X'\in \mathcal{X}\}$ and, for each $Y'\in \mathcal{Y}$, we check whether $Y'[S]\in \mathcal{X}_S$. If so, then we return YES.\@
    If processing all $Y'\in \mathcal{Y}$ is completed without a YES answer,
    then we return NO (for the given instance).\@

    If $\ED(X',Y')>\alpha$ holds for all $(X',Y')\in \mathcal{X}\times \mathcal{Y}$,
    then, by the union bound, the probability that $Y'[S]\in \mathcal{X}_S$ holds for some $Y'\in \mathcal{Y}$ is at most $\delta$.
    Thus, the algorithm returns YES with probability at most $\delta$.
    On the other hand, if $X'=Y'$ holds for some $(X',Y')\in \mathcal{X}\times \mathcal{Y}$,
    then $X'[S]=Y'[S]$, and we do return YES due to $Y'[S]\in \mathcal{X}_S$.

    If $\mathcal{X}_S$ is implemented as a ternary trie~\cite{BS79}, then its construction cost is \[\Oh\left(|\mathcal{X}|\left(\log |\mathcal{X}|+\tfrac{n \log \frac{n}{\delta}}{1+\alpha}\right)\right) = \Oh\left(|\mathcal{X}|\cdot \tfrac{n \log \frac{n}{\delta}}{1+\alpha}\right) =\Oh\left(\tfrac{\min\left(1+\beta,\sqrt{q(q+\beta)}\right)}{1+\alpha}\cdot n\log \tfrac{n}{\delta}\right).\]
    The time complexity of the second step is \[\Oh\left(|\mathcal{Y}|\left(\log|\mathcal{X}|+\tfrac{n\log \frac{n}{\delta}}{1+\alpha}\right)\right)=\Oh\left(\tfrac{\sqrt{q+\beta}}{(1+\alpha)\sqrt{q}}\cdot n\log \tfrac{n}{\delta}\right)\] per instance and $\Oh\big(\frac{\sqrt{q(q+\beta)}}{1+\alpha}\cdot n\log \frac{n}{\delta}\big)$ in total.
\end{proof}

\subsection{\GED for $h=1$}
\begin{lemma}\label{lem:g2}
    There exists a non-adaptive algorithm that, given a parameter $\delta\in \R_+$ and a batch of $q$ instances of \GEDa{\alpha}{\beta} satisfying $\beta^2 \le \frac{\alpha}{336\ceil{\log n}}$,
    solves the instances in $\Oh\big(\frac{\sqrt{q(q+\beta)}}{1+\alpha}\cdot n\log^2 n\cdot \log\frac{1}{\delta}\big)$ time with   each answer correct with probability at least $1-\delta$. Moreover, at most $\Oh\big(\frac{1+\beta}{1+\alpha}\cdot n\log^2 n \cdot \log \frac{1}{\delta})$ characters of the common string $X$ are accessed.
    \end{lemma}
\begin{proof}
    Let us assume that $\delta > \frac{1}{e}$; in general, we amplify the success probability by repeating the algorithm $\Oh(\log\frac{1}{\delta})$ times.
    If $\beta = 0$, then we simply use \cref{fct:hm}.
    Otherwise, we apply \cref{lem:gs} with $\phi=\beta$ and the algorithm of \cref{lem:s1} (with parameter $\Theta(\frac{1}{n})$) as the oracle.
    This is valid because \[\psi = \left\lfloor{\tfrac{112\beta^2\ceil{\log n}}{\alpha}}\right\rfloor \le \left\lfloor{\tfrac{112\alpha\ceil{\log n}}{336\alpha \ceil{\log n}}}\right\rfloor=\left\lfloor{\tfrac{1}{3}}\right\rfloor=0.\]
    Since the algorithm of \cref{lem:gs} is non-adaptive, the queries remain batched.
    The total running time is
    \[\Oh\left(\tfrac{\phi \log n}{\alpha}\cdot \tfrac{\sqrt{q(q+\beta)}}{1+\phi}\cdot n\log n\right)=\Oh\left(\tfrac{\sqrt{q(q+\beta)}}{1+\alpha} \cdot n\log^2 n\right),\]
    whereas the number of accessed characters of the common string $X$ does not exceed
    \[\Oh\left(\tfrac{\phi \log n}{\alpha}\cdot \tfrac{1+\beta}{1+\phi}\cdot n\log n\right)= \Oh\left(\tfrac{1+\beta}{1+\alpha}\cdot n\log^2 n\right).\qedhere\]
\end{proof}

\subsection{\SED for $h=1$}
\begin{lemma}\label{lem:s2}
    There exists a non-adaptive algorithm that, given a parameter $\delta\in \R_+$ and a batch of $q$ instances of  \SEDa satisfying $\gamma^2 \le \frac{\alpha}{3024\ceil{\log n}}$, solves the instances in \[\Oh\left(\left(\tfrac{\sqrt{q(q+\beta)}}{1+\alpha}+\tfrac{q(1+\beta)}{1+\alpha}\cdot \sqrt{\tfrac{\log n}{1+\alpha}}\right) n\log^2n\cdot \log\tfrac{n}{\delta}\right)\] time, using $\Oh\big(\tfrac{\sqrt{q(q+\beta)}}{1+\alpha}\cdot n\log^{2}n \cdot \log\tfrac{n}{\delta}\big)$ queries, and with each answer correct with probability at least $1-\delta$.
    Moreover, at most $\Oh\big(\tfrac{1+\beta}{1+\alpha}\cdot n\log^{2}n \cdot \log\tfrac{n}{\delta}\big)$  characters of the common string $X$ are accessed.
    \end{lemma}
\newcommand{\bgamma}{\bar{\gamma}}
\begin{proof}
    If $\gamma=0$, then we simply use \cref{lem:s1}.
    Consequently, we henceforth assume $\alpha\ge \beta \ge \gamma > 0$.

    In the remaining case, we proceed as in the proof of \cref{lem:sg} except that we artificially increase $\gamma$ to \[\bgamma:=\min\left(\beta, \left\lfloor{\sqrt{\tfrac{\alpha}{3024\ceil{\log n}}}}\right\rfloor\right),\]
    set $\xi = \max\big(\bgamma,\min\big(\beta, \big\lfloor{\tfrac{\bgamma\sqrt{\beta}}{\sqrt{q}}}\big\rfloor\big)\big)$, and use the algorithm of \cref{lem:g2} (with parameter $\Theta(\frac{\delta}{n})$) as the oracle; this is valid due to $(3\bgamma)^2 \le \frac{\alpha}{336\ceil{\log n}}$.
    The input instances are solved using $\Oh\big(\tfrac{\beta}{\xi}\big)$ batches of $\Oh(\tfrac{q\xi}{\bgamma})$ oracle calls,
    and these batches only differ in the strings $X'$ (common to each batch).

    If $q \le \frac{\bgamma^2}{\beta}$, then $\xi = \beta$, so the input instances are solved using $\Oh(1)$ batches of 
    $\Oh(\frac{q\beta}{\bgamma})=\Oh(\bgamma)$ oracle calls, and hence the total running time and the query complexity are
    \[
      \Oh\left(\tfrac{\sqrt{\frac{q\beta}{\bgamma}\cdot\bgamma}}{\alpha}\cdot n\log^2 n\cdot \log \tfrac{n}{\delta}\right)
    =\Oh\left(\tfrac{\sqrt{q\beta}}{\alpha}\cdot n\log^{2} n\cdot  \log \tfrac{n}{\delta}\right).
    \]
    If $q > \frac{\bgamma^2}{\beta}$, then the input instances are solved using $\Oh\big(\tfrac{\beta}{\xi}\big)$ batches of $\Theta(\tfrac{q\xi}{\bgamma})=\Omega(\tfrac{\xi \bgamma}{\beta})=\Omega(\bgamma)$ oracle calls.
    Hence, the total running time is
    \begin{multline*}
      \Oh\left(\tfrac{\beta}{\xi}\cdot \tfrac{q\xi}{\bgamma} \cdot \tfrac{1}{\alpha}\cdot n\log^2 n\cdot \log \tfrac{n}{\delta}\right)
    =\Oh\left(\tfrac{q\beta}{\alpha\bgamma}\cdot n\log^{2} n\cdot  \log \tfrac{n}{\delta}\right)\\
    =\Oh\left(\left(\tfrac{q}{\alpha}+\tfrac{q\beta}{\alpha}\cdot \sqrt{\tfrac{\log n}{\alpha}}\right)\cdot n\log^{2} n\cdot  \log \tfrac{n}{\delta}\right).
    \end{multline*}
    Moreover, the batches of oracle calls differ only by the strings $X'$ (common to each batch).
    Hence, the query complexity can be bounded as follows:
    \begin{multline*} \Oh\left(\left(\tfrac{\beta}{\xi}\cdot \tfrac{\bgamma}{\alpha}+ \tfrac{q\xi}{\bgamma} \cdot \tfrac{1}{\alpha}\right)\cdot n\log^2 n\cdot \log \tfrac{n}{\delta}\right)
      = \Oh\left(\left(\tfrac{\beta\gamma}{\alpha\beta} + \tfrac{\beta \bgamma \sqrt{q}}{\alpha \bgamma \sqrt{\beta}}+ \tfrac{q\bgamma \sqrt{\beta}}{\alpha \bgamma \sqrt{q}}+\tfrac{q\bgamma}{\alpha\bgamma}\right)\cdot n\log^2 n\cdot \log \tfrac{n}{\delta}\right) \\
      = \Oh\left(\left(\tfrac{\bgamma}{\alpha}+\tfrac{\sqrt{q \beta}}{\alpha}+\tfrac{q}{\alpha}\right)\cdot n\log^2 n\cdot \log \tfrac{n}{\delta}\right)
      =\Oh\left(\tfrac{\sqrt{q(q+ \beta)}}{\alpha}\cdot n\log^2 n\cdot \log \tfrac{n}{\delta}\right).
    \end{multline*}
    In any case, the number of accessed characters of the common string $X$ does not exceed
    \[
      \Oh\left(\tfrac{\beta}{\xi}\cdot \tfrac{\bgamma}{\alpha}\cdot n\log^2 n\cdot \log \tfrac{n}{\delta}\right)
      =\Oh\left(\tfrac{\beta}{\alpha}\cdot  n\log^2 n\cdot  \log \tfrac{n}{\delta}\right).\qedhere
    \]
\end{proof}

\subsection{\GED for $h=2$}
\begin{lemma}\label{lem:g3}
    There exists a non-adaptive algorithm that, given a parameter $\delta\in \R_+$ and a batch of $q$ instances of \SEDa satisfying $\beta \le \frac{\alpha^{2/3}}{336\ceil{\log n}}$, solves the instances in 
    \[\Oh\left(\left(\tfrac{\sqrt{q(q+\beta)}}{1+\alpha}+\tfrac{q(1+\beta)^2}{(1+\alpha)^2}\cdot \log^2 n\right)\cdot n\log^4 n \cdot \log\tfrac{1}{\delta}\right)\] time, using $\Oh\big(\frac{\sqrt{q(q+\beta)}}{1+\alpha}\cdot n\log^4 n \cdot \log\frac{1}{\delta}\big)$ queries, and with each answer correct with probability at least $1-\delta$.     Moreover, at most $\Oh\big(\tfrac{1+\beta}{1+\alpha}\cdot n\log^{4}n \cdot \log\tfrac{1}{\delta}\big)$  characters of the common string $X$ are accessed.
    \end{lemma}
\begin{proof}
    Let us assume that $\delta > \frac{1}{e}$; in general, we amplify the success probability by repeating the algorithm $\Oh(\log\frac{1}{\delta})$ times.
    If $\beta^2 \le \frac{\alpha}{336\ceil{\log n}}$, then we simply use \cref{lem:g2}.
    Otherwise, we apply \cref{lem:gs} with $\phi = \left\lfloor{\frac{\alpha^2}{\beta^2(336\ceil{\log n})^3}}\right\rfloor$ and the algorithm of \cref{lem:s2} (with parameter $\Theta(\frac{1}{n})$) as the oracle.
    This is valid due to the following inequalities:
    \begin{align*} 
    \psi &\le \tfrac{112\beta\phi\ceil{\log n}}{\alpha} \le \tfrac{112\beta \alpha^2\ceil{\log n}}{\alpha \beta^2 (336\ceil{\log n})^3} = \tfrac{\alpha}{3\beta(336\ceil{\log n})^2} < \tfrac{336\beta^2 \ceil{\log n}}{3\beta(336\ceil{\log n})^2}=\tfrac{\beta}{1008\ceil{\log n}} < \beta,\\
    \psi^2 &\le \tfrac{(112\beta\phi\ceil{\log n})^2}{\alpha^2} \le \tfrac{\phi\cdot (112\beta\ceil{\log n})^2\cdot \alpha^2}{\alpha^2 \cdot \beta^2(336\ceil{\log n})^3} = \tfrac{\phi}{3024\ceil{\log n}},\\
    \phi &= \left\lfloor\tfrac{\alpha^2}{\beta^2(336 \ceil{\log n})^3}\right\rfloor \ge \left\lfloor\tfrac{(336\beta\ceil{\log n})^3}{(336\beta \ceil{\log n})^3}\right\rfloor = \beta.\end{align*}
    Since the algorithm of \cref{lem:gs} is non-adaptive, the queries remain batched.

    The total running time is 
    \begin{multline*}\Oh\left(\tfrac{\phi \log n}{\alpha}\cdot \left(\tfrac{\sqrt{q(q+\beta)}}{\phi}+\tfrac{q\beta}{\phi}\cdot \sqrt{\tfrac{\log n}{\phi}}\right)\cdot n\log^3 n \right)
    =\Oh\left(\left(\tfrac{\sqrt{q(q+\beta)}}{\alpha} +\tfrac{q\beta\sqrt{\log n}}{\alpha\sqrt{\phi}}\right)\cdot n \log^{4} n\right)\\
    =\Oh\left(\left(\tfrac{\sqrt{q(q+\beta)}}{\alpha}+\tfrac{q\beta^2}{\alpha^2}\cdot \log^2 n\right)\cdot n\log^4 n\right),\end{multline*}
    whereas the query complexity is 
    \[\Oh\left(\tfrac{\phi \log n}{\alpha}\cdot \tfrac{\sqrt{q(q+\beta)}}{\phi}\cdot n\log^3 n\right)=\Oh\left(\tfrac{\sqrt{q(q+\beta)}}{\alpha}\cdot n\log^4 n\right).\]
    Moreover, the number of accessed characters of the common string $X$ does not exceed
    \[\Oh\left(\tfrac{\phi \log n}{\alpha}\cdot \tfrac{\beta}{\phi}\cdot n \log^3 n\right)=\Oh\left(\tfrac{\beta}{\alpha}\cdot n \log^4 n\right).\qedhere\]

\end{proof}

\subsection{\SED for $h=2$}
\begin{lemma}\label{lem:s3}
    There exists a non-adaptive algorithm that, given a parameter $\delta\in \R_+$ and an instance of \SEDa satisfying $\gamma \le \frac{\alpha^{2/3}}{1008\ceil{\log n}}$, solves the instance in \[\Oh\left(\left(\tfrac{\sqrt{1+\beta}}{1+\alpha}+\tfrac{(1+\beta)(1+\gamma)}{(1+\alpha)^2}\cdot \log^2 n\right)\cdot n\log^4 n \cdot \log\tfrac{n}{\delta}\right)\] time, using $\Oh\big(\frac{\sqrt{1+\beta}}{1+\alpha}\cdot n\log^{4} n \cdot \log\frac{n}{\delta}\big)$ queries, and with error probability at most $\delta$.
    \end{lemma}
\begin{proof}
    If $\gamma^2 \le \frac{\alpha}{3024\ceil{\log n}}$, then we simply use \cref{lem:s2}.
    In this case, due to $\alpha \ge \beta$, the running time is 
    \[
      \Oh\left(\left(\tfrac{\sqrt{1+\beta}}{1+\alpha} + \tfrac{1+\beta}{1+\alpha}\cdot \sqrt{\tfrac{\log n}{1+\alpha}}\right)\cdot n\log^2n \cdot \log\tfrac{n}{\delta}\right)
      =\Oh\left(\tfrac{\sqrt{1+\beta}}{1+\alpha}\cdot n\log^{2.5} n \cdot \log\tfrac{n}{\delta}\right).
   \]
    Otherwise, we proceed as in the proof of \cref{lem:sg} except that we set $\xi=\min(\beta,\floor{\gamma \sqrt{\beta}})$ and use \cref{lem:g3} (with parameter $\Theta(\frac{\delta}{n})$) as the oracle (this is valid due to $3\gamma \le \frac{\alpha^{2/3}}{336\ceil{\log n}}$).
    The input instances are solved using $\Oh\big(\tfrac{\beta}{\xi}\big)$ batches of $\Oh(\tfrac{\xi}{\bgamma})$ oracle calls,
    and these batches only differ in the strings $X'$ (common to each batch).

    If $\beta \le \gamma^2$, then $\xi = \beta$, so the input instances are solved using $\Oh(1)$ batches of $\Oh(\frac{\beta}{\gamma})=\Oh(\gamma)$ queries, and the hence the running time is
\[      \Oh\left(\left(\tfrac{\sqrt{\frac{\beta}{\gamma}\cdot\gamma}}{\alpha}+\tfrac{\frac{\beta}{\gamma}\cdot \gamma^2}{\alpha^2}\cdot \log^2 n\right)\cdot n\log^4 n\cdot \log \tfrac{n}{\delta}\right)
    =\Oh\left(\left(\tfrac{\sqrt{\beta}}{\alpha}+\tfrac{\beta \gamma}{\alpha^2}\cdot \log^2 n \right)\cdot n\log^{2} n\cdot  \log \tfrac{n}{\delta}\right),
   \]
    whereas the query complexity is
    \[
      \Oh\left(\tfrac{\sqrt{\frac{\beta}{\gamma}\cdot\gamma}}{\alpha}\cdot n\log^4 n\cdot \log \tfrac{n}{\delta}\right)
    =\Oh\left(\tfrac{\sqrt{\beta}}{\alpha}\cdot n\log^4 n\cdot  \log \tfrac{n}{\delta}\right).
    \]
     Otherwise, $\xi = \gamma \sqrt{\beta}$, so the input instances are solved using $\Oh(\tfrac{\sqrt{\beta}}{\gamma})$ batches
    of $\Theta(\sqrt{\beta})=\Omega(\gamma)$ queries, and hence the running time is
    \begin{multline*}
      \Oh\left(\tfrac{\sqrt{\beta}}{\gamma}\cdot \left(\tfrac{\sqrt{\beta}}{\alpha}+\tfrac{\sqrt{\beta}\cdot \gamma^2}{\alpha^2}\cdot \log^2 n\right)\cdot n\log^4 n\cdot \log \tfrac{n}{\delta}\right)
    =\Oh\left(\left(\tfrac{\beta\gamma}{\alpha \gamma^2}+\tfrac{\beta \gamma}{\alpha^2}\cdot \log^2 n \right)\cdot n\log^4 n\cdot  \log \tfrac{n}{\delta}\right)\\
    =\Oh\left(\left(\tfrac{\beta\gamma}{\alpha \frac{\alpha}{\log n}}+\tfrac{\beta \gamma}{\alpha^2}\cdot \log^2 n \right)\cdot n\log^4 n\cdot  \log \tfrac{n}{\delta}\right)
    =\Oh\left(\tfrac{\beta \gamma}{\alpha^2}\cdot n\log^6 n\cdot  \log \tfrac{n}{\delta}\right),
    \end{multline*}
    where the lower bound on $\gamma$ is due to $\gamma^2 > \tfrac{\alpha}{3024\ceil{\log n}}$.
    Moreover, the batches of oracle calls differ only by the strings $X'$ (common to each batch).
    Hence, the query complexity can be bounded as follows:
    \[
      \Oh\left(\left(\tfrac{\sqrt{\beta}}{\gamma}\cdot \tfrac{\gamma}{\alpha} + \tfrac{\sqrt{\beta}}{\alpha}\right)\cdot n\log^4 n\cdot \log \tfrac{n}{\delta}\right) 
    =\Oh\left(\tfrac{\sqrt{\beta}}{\alpha}\cdot n\log^4 n\cdot  \log \tfrac{n}{\delta}\right).\qedhere
    \]
    \end{proof}

\subsection{\GED and \SED for $h\ge 3$}

\begin{theorem}\label{thm:main}
    There exists a non-adaptive algorithm that, given $h\in \Zz$, $\delta\in \R_+$,
    and an instance of \GEDa satisfying $\beta < (336\ceil{\log n})^{\frac{-h}{2}} \alpha^{\frac{h}{h+1}}$,
    solves the instance in \[\Oh\left(\left(\tfrac{\sqrt{1+\beta}}{1+\alpha}+\tfrac{(1+\beta)^2}{(1+\alpha)^2}\cdot \log^h n \right)\cdot n\log^{2h}n \cdot \log\tfrac{1}{\delta} \cdot 2^{\Oh(h)}\right)\] time,
    using $\Oh\big(\frac{\sqrt{1+\beta}}{1+\alpha}\cdot n\log^{2h}n\cdot \log\tfrac{1}{\delta}  \cdot 2^{\Oh(h)}\big)$ queries, and with error probability at most $\delta$.
\end{theorem}
\begin{theorem}\label{thm:main2}
    There exists a non-adaptive algorithm that, given $h\in \mathbb{Z}_{\ge 2}$, $\delta\in \R_+$,
    and an instance of \SEDa satisfying $\gamma < \frac13(336\ceil{\log n})^{\frac{-h}{2}} \alpha^{\frac{h}{h+1}}$, solves the instance in 
    \[\Oh\left(\left(\tfrac{\sqrt{1+\beta}}{1+\alpha}+\tfrac{(1+\beta)(1+\gamma)}{(1+\alpha)^2}\cdot \log^h n\right)\cdot n\log^{2h} n\cdot \log\tfrac{n}{\delta}  \cdot 2^{\Oh(h)}\right)\] time,
    using $\Oh\big(\frac{\sqrt{1+\beta}}{1+\alpha}\cdot n\log^{2h} n\cdot \log\tfrac{n}{\delta}  \cdot 2^{\Oh(h)}\big)$ queries, and with error probability at most $\delta$.
\end{theorem}
\begin{proof}[Proof of \cref{thm:main,thm:main2}]
    As for \GEDa, let us assume that $\delta > \frac{1}{e}$; in general, we amplify the success probability by repeating the algorithm $\Oh(\log\frac{1}{\delta})$ times.
    We use \cref{fct:hm,lem:g2,lem:g3} when applicable.
    In particular, this covers $\beta \le \frac{\alpha^{2/3}}{336\ceil{\log n}}$ and $h\le 2$.
    Otherwise, we apply \cref{lem:gs} with $\phi = \beta$ using our algorithm for \SEDa{\phi}{\beta}{\psi} (with parameters $h-1$ and $\Theta(\frac{1}{n})$) as the oracle.
    This is valid because
\[\psi \le \tfrac{112 \beta^2 \ceil{\log n}}{\alpha} < \tfrac{112 \beta^2\ceil{\log n}}{\beta^{\frac{h+1}{h}}\cdot (336 \ceil{\log n})^{\frac{h+1}{2}}} =
\tfrac13\cdot \left(336 \ceil{\log n}\right)^{\frac{1-h}{2}}\cdot \phi^{\frac{h-1}{h}}.\]
The running time is
\begin{multline*}
\Oh\left(\tfrac{\phi \log n}{\alpha} \cdot \left(\tfrac{\sqrt{\beta}}{\phi}+\tfrac{\psi\beta}{\phi^2}\cdot \log^{h-1}  n\right)\cdot n \log^{2h-1} n \cdot 2^{\Oh(h-1)}\right)\\
= \Oh\left(\left(\tfrac{\sqrt{\beta}}{\alpha}+\tfrac{\beta^2}{\alpha^2}\cdot \log^{h} n\right) \cdot n\log^{2h}n \cdot 2^{\Oh(h)}\right),
\end{multline*}
whereas the query complexity is 
\[\Oh\left(\tfrac{\phi \log n}{\alpha} \cdot \tfrac{\sqrt{\beta}}{\phi}\cdot n\log^{2h-1} n \cdot 2^{\Oh(h-1)}\right)
= \Oh\left(\tfrac{\sqrt{\beta}}{\alpha}\cdot n\log^{2h} n\cdot 2^{\Oh(h-1)}\right).
\]

As for \SEDa, we use \cref{lem:s1,lem:s2,lem:s3} when applicable.
In particular, this covers $\gamma \le \frac{\alpha^{2/3}}{1008\ceil{\log n}}$ and $h\le 2$.
Otherwise, we apply \cref{lem:sg} using our algorithm for \GEDa{\alpha}{3\gamma} (with parameters $h$ and $\Theta(\frac{\delta}{n})$) as the oracle.
This is valid because $3\gamma \le (336\ceil{\log n})^{\frac{-h}2} \alpha^{\frac{h}{h+1}}\le \alpha$.

The running time is \[
\Oh\left(\tfrac{\beta}{\gamma} \cdot \left(\tfrac{\sqrt{\gamma}}{\alpha}+\tfrac{\gamma^2}{\alpha^2}\cdot \log^{h}n\right)\cdot n\log^{2h} n\cdot  \log\tfrac{n}{\delta}\cdot 2^{\Oh(h)}\right)=
\Oh\left(\tfrac{\beta}{\gamma} \cdot\tfrac{\gamma^2}{\alpha^2}\cdot n\log^{3h}n \cdot\log\tfrac{n}{\delta}\cdot 2^{\Oh(h)}\right),\]
whereas the query complexity is
\[
\Oh\left(\tfrac{\sqrt{\beta}}{\sqrt{\gamma}} \cdot \tfrac{\sqrt{\gamma}}{\alpha}\cdot n\log^{2h}n \cdot\log\tfrac{n}{\delta}\cdot 2^{\Oh(h)}\right)
= \Oh\left(\tfrac{\sqrt{\beta}}{\alpha}\cdot n\log^{2h}n\cdot\log\tfrac{n}{\delta} \cdot 2^{\Oh(h)}\right).\qedhere\]
\end{proof}

\section{Matching Lower Bound for Non-Adaptive Query Complexity}\label{sec:lb}
\newcommand{\DY}{\mathcal{D}_{\mathrm{YES}}}
\newcommand{\DN}{\mathcal{D}_{\mathrm{NO}}}

In this section we strengthen the following lower bound of~\cite{BEKMRRS03}
for the \GEDa{\frac{n}{6}}{\beta} problem.
 \begin{proposition}[\cite{BEKMRRS03}]\label{prp:lb}
For all integers $n,\alpha,\beta\in \Zp$ such that $\frac{n}{6} = \alpha \ge \beta$,
every algorithm solving all instances of the \GEDa problem has worst-case query complexity
$\Omega(\sqrt{\beta})$ or error probability exceeding $\frac13$.
\end{proposition}

\begin{theorem}\label{thm:lb}
    For all integers $n,\alpha,\beta\in \Zp$ such that $\frac{n}{6} \ge \alpha \ge \beta$,
    every non-adaptive algorithm solving all instances the \GEDa problem has expected query complexity $\Omega(\frac{n\sqrt{\beta}}{\alpha})$ or error probability exceeding $\frac13$.
\end{theorem}
\begin{proof}
    Suppose that, for some fixed $n,\alpha,\beta\in \Zp$ with $\frac{n}{6} \ge \alpha \ge \beta$,
    there exists a non-adaptive algorithm $A$ that uses $q$ queries in expectation and errs with probability at most $\frac13$. We shall derive an algorithm $A'$ for $n=6\alpha$ that uses $\frac{486q\alpha}{n}$ queries in the worst case; if $q=o(\frac{n\sqrt{\beta}}{\alpha})$, this would contradict \cref{prp:lb}.

    Let us first define an algorithm $A^3$ that runs $A$ three times and returns the dominant answer;
    it has error probability to $\frac{1+3\cdot 2}{27}=\frac{7}{27}$
    and expected query complexity to~$3q$.
    For each $i\in [0\dd \floor{\frac{n}{6\alpha}})$, let $q_i$ be the expected number of queries
    that $A^3$ makes to $X[6\alpha i \dd 6\alpha(i+1))$ and $Y[6\alpha i \dd 6\alpha(i+1))$;
    since $A^3$ is non-adaptive, these values do not depend on $X$ or $Y$.
    By linearity of expectation, we have $\sum_{i} q_i \le 3q$, so there exists $i\in [0\dd \floor{\frac{n}{6\alpha}})$ such that 
    $q_i \le \frac{3q}{\floor{\frac{n}{6\alpha}}} \le \frac{36q\alpha}{n}$.

    The algorithm $A'$, given strings $X',Y'\in \Sigma^{6\alpha}$,
    constructs strings $X=\mathtt{a}^{6\alpha i} \cdot X' \cdot \mathtt{a}^{n-6\alpha(i+1)}$
    and $Y=\mathtt{a}^{6\alpha i} \cdot Y' \cdot \mathtt{a}^{n-6\alpha(i+1)}$, where $\mathtt{a}\in \Sigma$ is an arbitrary character.
    Formally, this means that an oracle providing random access to $(X',Y')$ is transformed 
    into an oracle providing random access to $(X,Y)$.
    Then, $A'$ runs $A^3(X,Y)$, but it terminates the execution (returning an arbitrary answer)
    on an attempt to make more than $\frac{27}{2}q_i \le \frac{486q\alpha}{n}$ queries to $(X',Y')$.
    
    This cap of the number of queries trivially bounds the query complexity of $A'$.
    As for the correctness, observe that \cref{fct:hered} implies $\ED(X,Y)=\ED(X',Y')$.
    Moreover, there is a one-to-one correspondence between the queries of $A'(X',Y')$
    and the queries that $A^3(X,Y)$ makes to $X[6\alpha i \dd 6\alpha(i+1))$ and $Y[6\alpha i \dd 6\alpha(i+1))$.
    Hence, it suffices to analyze the error probability of the capped version of $A^3$.
    Without the query limit, $A^3(X,Y)$ would in expectation make $q_i$ queries to $(X',Y')$ and err with probability at most $\frac{7}{27}$.
    By Markov's inequality, the probability of making more than $\frac{27}{2}q_i$ queries to $(X',Y')$
    does not exceed $\frac{2}{27}$. Thus, the execution of $A^3$ is terminated with probability at most $\frac{2}{27}$. Overall, this increases the error probability from $\frac{7}{27}$ to $\frac{7+2}{27}=\frac{1}{3}$.
\end{proof}

{\small
\bibliographystyle{alphaurl}
\bibliography{references}
}

\end{document}